\documentclass[11pt]{article}

\usepackage{graphicx}
\usepackage{amsmath}
\usepackage{amssymb}
\usepackage{xcolor}
\usepackage[lined,ruled,vlined]{algorithm2e}
\usepackage{thm-restate,xspace}

\usepackage{amsthm}
\usepackage{geometry}
\usepackage[compact]{titlesec}

\usepackage[breaklinks]{hyperref}
\hypersetup{colorlinks=true, 
            citebordercolor={.6 .6 .6},linkbordercolor={.6 .6 .6},%
citecolor=blue,urlcolor=black,linkcolor=blue,pagecolor=black}
\usepackage[nameinlink]{cleveref}
\usepackage{enumitem}

\makeatletter
\allowdisplaybreaks
\makeatother

\newtheorem{theorem}{Theorem}[section]
\newtheorem{assumption}[theorem]{Assumption}
\newtheorem{lemma}[theorem]{Lemma}
\newtheorem{definition}[theorem]{Definition}

\newtheorem{corollary}[theorem]{Corollary}

\newtheorem{observation}[theorem]{Observation}

\newcommand{\ignore}[1]{}

\newcommand{\ext}{{\tt Extend}}

\newcommand{\I}{\ensuremath{\mathcal{I}}\xspace}
\newcommand{\detc}{{\textsc{DetCost}}\xspace}
\newcommand{\makespan}{\textsc{GenMakespan}\xspace}
\newcommand{\smm}{{\makespan}\xspace}
\def\al{\alpha}

\def\J{J}

\def\cI{{\cal I}}
\def\cL{{\cal L}}
\def\Ber{{\tt Ber}}
\def\cP{{\cal P}}
\def\cS{{\cal S}}

\def\packable{{packable}\xspace}
\def\safe{{safe}\xspace}
\def\sse{\subseteq}
\def\opt{\ensuremath{\mathsf{OPT}}\xspace}
\def\E{\mathbb{E}}
\def\R{\mathbb{R}}
\def\P{\mathbb{P}}

\def\sse{\subseteq}

\def\cI{\ensuremath{{\mathcal{I}}}\xspace}

\def\dI{\ensuremath{{\mathcal{C}}}\xspace}
\def\UFP{{\tt UFP-Tree}\xspace}

\newcommand{\initOneLiners}{%
    \setlength{\itemsep}{0pt}
    \setlength{\parsep }{0pt}
    \setlength{\topsep }{0pt}
}
\newenvironment{OneLiners}[1][\ensuremath{\bullet}]
    {\begin{list}
        {#1}
        {\initOneLiners}}
    {\end{list}}

\author{Anupam Gupta\thanks{Computer Science Department, Carnegie
    Mellon University, Pittsburgh, USA. Supported in part by NSF award     CCF-1907820, CCF-1955785, and CCF-2006953, and by the Indo-US Joint Center for Algorithms Under Uncertainty. } \and Amit Kumar\thanks{Dept. of Computer Science and Engg., IIT Delhi,
		India 110016.} \and Viswanath Nagarajan\thanks{Department of Industrial and Operations
		Engineering, University of Michigan, Ann Arbor, MI 48109. Supported in part by NSF grants  CCF-1750127, CMMI-1940766, and CCF-2006778. }
	\and Xiangkun Shen\thanks{Yahoo! Research, New York, NY 10003. The work was partially done when the author was a student at Department of Industrial and Operations Engineering, University of Michigan.} }

\title{Stochastic Makespan Minimization in Structured Set Systems}

\date{\today}
\begin{document}

\maketitle

\begin{abstract}
  We study stochastic combinatorial optimization problems where the
  objective is to minimize the expected maximum load (a.k.a.\ the
  makespan). In this framework, we have a set of $n$ tasks and $m$
  resources, where each task $j$ uses some subset of the resources.
  Tasks have random sizes $X_j$, and our goal is to non-adaptively
  select $t$ tasks to minimize the expected maximum load over all
  resources, where the load on any resource $i$ is the total size of
  all selected tasks that use $i$.
  For example, when resources are points and tasks are intervals in a line, we obtain an
  $O(\log\log m)$-approximation algorithm. 
   Our technique is also applicable to other
  problems with some geometric structure in the relation
  between tasks and resources; e.g., packing paths, rectangles, and
  ``fat'' objects.  
  Our approach uses a strong LP relaxation using the cumulant
  generating functions of the random variables. We also show that this
  LP has an $\Omega(\log^* m)$ integrality gap, even for the problem of
  selecting intervals on a line; here $\log^* m$ is the iterated
  logarithm function. 
 
\end{abstract}

\section{Introduction}\label{sec:intro}
Consider the following task scheduling problem: an event center
receives requests/tasks from its clients. Each task $j$ specifies a
start and end time (denoted $(a_j, b_j)$), and the amount $x_j$ of
some shared resource (e.g., staff support) that this task requires
throughout its duration.  The goal is to accept some target $t$ number
of tasks so that the maximum resource-utilization over time is as
small as possible. Concretely, we want to choose a set $S$ of tasks
with $|S|=t$ to minimize
\begin{gather*}
  \max_{\text{times } \tau} \underbrace{\sum_{j \in S: \tau \in [a_j,
      b_j]} x_j}_{\text{usage at time $\tau$}}~~.
\end{gather*}
  This can be modeled as an interval packing problem: if the
sizes are identical, the natural LP is totally unimodular and we get
an exact algorithm. For general sizes, there is a constant-factor
approximation algorithm~\cite{chakrabarti2007approximation}.

However, in many settings, we may not know the resource consumption
$X_j$ precisely up-front, at the time we need to make a
decision. Instead, we may be only given estimates. What if the
requirement $X_j$ is a random variable whose distribution is given to
us? Again we want to choose $S$ of size $t$, but this time we want to
minimize the \emph{expected} maximum usage:
\begin{gather*}
  \E\bigg[ \max_{\text{times } \tau} \sum_{j \in S: \tau \in [a_j,
    b_j]} X_j \bigg].  
\end{gather*}
Note that our decision to pick task $j$ affects all times in $[a_j,
b_j]$, and hence the loads on various places are no longer
independent: how can we effectively reason about such a problem?

In this paper we consider general resource allocation problems of the
following form. There are several tasks and resources, where each task
$j$ has some size $X_j$ and uses some subset $U_j$ of resources. That
is, if task $j$ is selected then it induces a load of $X_j$ on every
resource in $U_j$. Given a target $t$, we want to select a subset $S$
of $t$ tasks to minimize the \emph{expected maximum load} over all
resources.  For the non-stochastic versions of these problems (when
$X_j$ is a single value and not a random variable), we can use the
natural linear programming (LP) relaxation and randomized rounding to
get an $O(\frac{\log m}{\log\log m})$-approximation algorithm~\cite{ChekuriVZ10}; here
$m$ is the number of resources.  However, much better results are known
when the task-resource incidence matrix has some geometric structure.
One such example appeared above: when the resources have some linear
structure, and the tasks are intervals. Other examples include
selecting rectangles in a plane (where tasks are rectangles and
resources are points in the plane), and selecting paths in a tree
(tasks are paths and resources are edges/vertices in the tree). This
class of problems has received a lot of attention and has strong
approximation guarantees, see
e.g. \cite{Chan04,AgarwalM06,ChekuriMS07,ChanH12,ChalermsookC09,ChalermsookW21,chakrabarti2007approximation}.

However, the {\em stochastic} counterparts of these resource
allocation problems remain wide open. Can we achieve good
approximation algorithms when the task sizes $X_j$ are random
variables? We refer to this class of problems as \textrm{stochastic
  makespan minimization} (\smm). In the rest of this work, we assume
that the distributions of all the random variables are known, and that
the random variables $X_j$s are independent.

\subsection{Results and Techniques}
\label{sec:results-techniques}

We show that good approximation algorithms are indeed possible for
\smm problems that have certain geometric structure. We consider the
following two assumptions:
\begin{itemize}
\item {\em Deterministic problem assumption:} There is an LP-based $\alpha$-approximation algorithm for a  
  deterministic variant of \smm.
\item {\em Well-covered assumption:}  
  for any subset $D\sse [m]$ of resources and tasks $L(D)$ incident to
  $D$, the tasks in $L(D)$ incident to any resource $i\in [m]$ are
  ``covered'' by at most $\lambda$ resources in $D$.
\end{itemize}
These assumptions are formalized in \S\ref{sec:prdef}. To give some
intuition for these assumptions, consider intervals on the line. The
first assumption holds by the results
of~\cite{chakrabarti2007approximation}. The second assumption holds
because each resource is some time $\tau$, and the tasks using time
$\tau$ can be covered by two resources in $D$, namely the closest times
$\tau_1, \tau_2 \in D$ such that $\tau_1 \leq \tau \leq \tau_2$.

Our informal main result is the following:

\begin{theorem}[Main (Informal)]
  There is an $O(\alpha \lambda \log\log m)$-approximation algorithm for
  stochastic makespan minimization (\smm), with $\alpha$ and $\lambda$
  as in the above assumptions.
\end{theorem}

We also show that both $\alpha$ and $\lambda$ are small in a number
of geometric settings: for intervals on a line, for paths in a tree,
and for rectangles and ``fat objects'' in a plane. Therefore, we obtain 
$poly(\log\log m)$-approximation algorithms in all these cases. 

A first naive approach for \smm is (i) to write an LP relaxation with
expected sizes $\E[X_j]$ as deterministic sizes and then (ii) to use any
LP-based $\alpha$-approximation algorithm for the deterministic
problem. However, this approach only yields an
$O(\alpha \frac{\log m}{\log\log m})$ approximation ratio, due to the
use of union bounds in calculating the expected maximum. Our idea is to use the structure of the problem to improve the approximation ratio.

Our approach is as follows. First, we use the (scaled) logarithmic
moment generating function (log-mgf) of the random variables $X_j$ to
define deterministic surrogates to the random sizes. Second, we
formulate a strong LP relaxation with an exponential number of
``volume'' constraints that use the log-mgf values. These two ideas
were used earlier for 
 stochastic makespan minimization in settings where each task loads a
single resource~\cite{KRT,gupta2018stochastic}. In the example above,
this would handle cases where each task uses only a single time instant. However, we need a
more sophisticated LP for \smm to be able to handle the combinatorial
structure when tasks use many resources. Despite the large number of
constraints, this LP can be solved approximately in polynomial time,
using the ellipsoid method and using a maximum-coverage algorithm as
the separation oracle.  Third (and most important), we provide an
iterative-rounding algorithm that partitions the tasks/resources into
$O(\log\log m)$ many nearly-disjoint instances of the deterministic
problem. The analysis of our rounding algorithm relies on both the
assumptions above, and also on the volume constraints in our LP and on
properties of the log-mgf.

We also show some limitations of our approach. For \smm involving
intervals in a line (which is our simplest application), we prove that
the integrality gap of our LP is $\Omega(\log^* m)$. This rules out a
constant-factor approximation via this LP. For \smm on more general
set-systems (without any structure), we prove that the integrality gap
can be $\Omega(\frac{\log m}{(\log\log m)^2})$ even if all deterministic instances solved in our algorithm have an 
$\alpha=O(1)$ integrality gap. This suggests that we do need to exploit additional
structure---such as the well-covered assumption above---in order to
obtain significantly better approximation ratios via our LP.

\subsection{Related Work}
\label{sec:related-work}

The deterministic counterparts of the problems studied here are well-understood. In particular, there are very good LP-based approximation algorithms for maximum-weight packing of intervals in a line~\cite{chakrabarti2007approximation}, paths in a tree (with edge loads)~\cite{ChekuriMS07}, rectangles in a plane \cite{ChalermsookW21}  and  fat-objects in a plane~\cite{ChanH12}. 

Our techniques draw on prior work on stochastic makespan minimization
for identical~\cite{KRT} and unrelated~\cite{gupta2018stochastic}
resources; but there are also important new ideas. In particular, the
use of log-mgf values as the deterministic proxy for random variables
comes from \cite{KRT} and the use of log-mgf values at multiple scales
comes from \cite{gupta2018stochastic}. The ``volume'' constraints in
our LP also has some similarity to those in
\cite{gupta2018stochastic}: however, a key difference here is that the
random variables loading different resources are correlated (whereas
they were independent in \cite{gupta2018stochastic}).  Indeed, this is
why our LP can only be solved approximately whereas the LP relaxation
in \cite{gupta2018stochastic} was optimally solvable. We emphasize
that our main contribution is the rounding algorithm which uses a new
set of ideas; these lead to the $O(\log\log m)$ approximation bound,
whereas the rounding in \cite{gupta2018stochastic} obtained a
constant-factor approximation. We also prove  a super-constant integrality gap in our setting (even for  intervals in a line), which rules out the possibility of a constant-factor approximation via our LP.

The stochastic load balancing problem on unrelated resources has also been studied for general $\ell_p$-norms (note that the makespan corresponds to the $\ell_\infty$-norm) and a constant-factor approximation is known~\cite{Molinaro19}. We do  not consider $\ell_p$-norms in this paper. 

\section{Problem Definition and Preliminaries}
\label{sec:prdef}

We are given $n$ tasks and $m$ resources. Each task $j\in[n]$ uses
some subset $U_j\sse [m]$ of resources.  
For each resource $i\in[m]$, define $L_i\sse[n]$ to be the tasks that
utilize $i$. Each task $j\in [n]$ has a {\em random} size $X_j$. If a task $j$ is selected into our set $S$, it adds a
load of $X_j$ to each resource in $U_j$: the load on resource
$i\in [m]$ is $Z_i:=\sum_{j\in S\cap L_i}X_j$. The makespan is the
maximum load, i.e. $\max_{i=1}^m Z_i$.  The goal is to select a subset
$S\sse [n]$ with $t$ tasks to minimize the expected \emph{makespan}:
\begin{equation}\label{eq:def}
    \min_{S\sse [n]: |S|=t}\quad \E\bigg[ \max_{i=1}^m  \, \sum_{j\in S\cap L_i}X_j \bigg].
\end{equation}
The
distribution of each random variable (r.v.)   $X_j$ is known, and these distributions
are independent. 
We assume that all the $X_j$s are discrete r.v.s with polynomial support size. We also assume that  each    distribution is available explicitly (as a list of realizations and probabilities). 
In our algorithm, we will  use these distributions  to compute some ``effective'' sizes (defined in \S\ref{sec:effect-size-rand}).

For any subset $K\sse [m]$ of resources, let $L(K) :=\bigcup_{i\in K} L_i$ be the set of tasks that utilize at least one resource in $K$. 
  
\subsection{Structure of Set Systems: The Two Assumptions}
\label{sec:struct-set-syst}

Our results hold when the following two properties are satisfied by
the set system $([n], {\cL})$, where ${\cL}$ is the collection
of sets $L_i$ for each $i \in [m]$. Note that the set system has
$n$ elements (corresponding to tasks) and $m$ sets (corresponding to resources).

\begin{enumerate}[label=\textbf{A\arabic*}]
    \item \label{assump:int-gap} {\bf ($\alpha$-\packable)}: A set
      system $([n], \cL)$ is said to be {\em $\alpha$-\packable} if
      for any assignment of size $s_j \geq 0$ and reward $r_j$ to each
      element $j \in [n]$, and any threshold parameter $\theta \geq
      \max_j s_j$, there is a polynomial-time algorithm that rounds a
      fractional solution
      $y$ to the
      following LP relaxation into an integral solution $\widehat{y}$,
      losing a factor of at most $\alpha \geq 1$:   
    \begin{equation}\label{eq:det_LP}
    \max \bigg\{ \sum_{j \in [n]} r_j\cdot y_j \,:\, \sum_{j\in L}
    s_j\cdot y_j \le \theta, \, \, \forall L \in \cL; \   0\le y_j\le 1, \, \, \forall j \in [n] \bigg\}. 
    \end{equation}
That is, $\sum_j r_j
      \widehat{y}_j \geq \frac1\alpha \sum_j r_j
      y_j$. We also assume, without loss of generality, that the support of $\widehat{y}$ is contained in the support of $y$. (The support of vector $z\in \R^n_+$ is $\{j\in [n] : z_j>0\}$ which corresponds to its positive entries.) 

  \item \label{assump:net} {\bf ($\lambda$-\safe)}: Let $[m]$ be
    the indices of the sets in $\cL$; recall that these are  the resources.  
     The set system $([n], \cL)$ is
    {\em $\lambda$-safe} if there is a polynomial-time algorithm that, given any subset $D \sse [m]$ of
    (``dangerous'') resources, finds a subset $M \supseteq D$ of
    (``safe'') resources, such that 
    \begin{enumerate}
        \item $|M|$ is polynomially bounded by
    $|D|$, and
    \item for every $i \in [m]$, there is a subset
    $R_i \subseteq M$, $|R_i| \leq \lambda$, such that
    $L_i \cap L(D) \subseteq L(R_i)$; in other words, every task that uses $i$ and some resource from $D$ also uses a resource from $R_i$.  
    \end{enumerate}
    Recall that $L(D)=\bigcup_{h\in D}L_h$. We denote the set $M$
    as $\ext(D)$.
\end{enumerate}

Let us give an example. Suppose $P = [m]$ are $m$ points on the line,
and consider $n$ intervals $I_1, \ldots, I_n$ of the line with each
$I_j \sse P$. Now the set system is defined on $n$ elements (one for
each interval), with $m$ sets where set $L_i$ for point $i \in [m]$
consists of the indices of all intervals that contain $i$. The $\lambda$-safe
condition says that for any subset $D$ of points in $P$, we can find a
superset $M$ which is not much larger such that for any point $i$ on
the line, there are $\lambda$ points in $M$ containing all the intervals
that pass through both $i$ and $D$. In other words, if these intervals
contribute any load to $i$ and $D$, they also contribute to one of these
$\lambda$ points.  And indeed, choosing $M = D$ ensures that $\lambda=2$:
for any $i$ we choose the nearest points in $M$ on either side of
$i$.

Other families that are $\alpha$-packable and $\lambda$-safe include:
\begin{itemize}
\item Each element in $[n]$ corresponds to a path in a tree, with the set
  $L_i$ being the subset of paths through node $i$. See Lemmas~\ref{lem:UFtree_set}-\ref{lem:UFtree_LP} for the proof.  
\item Elements in $[n]$ correspond to rectangles or fat-objects in a plane, and each $L_i$
  consists of the elements containing a particular point $i$ in the plane. See Lemmas~\ref{lem:RecPack_set} and \ref{lem:RecFat_set}.
\end{itemize}    

For a subset $X \sse [n]$, the projection of $([n], \cL)$ to $X$ is the
smaller set system $(X, \cL|_X)$, where $\cL|_X = \{L \cap X \mid L \in \cL\}$.
Loosely speaking, the following lemma formalizes that   packability and safeness properties also hold  for sub-families and 
 disjoint unions.

\begin{restatable}{lemma}{Assumption}
 Consider a   set system $([n], \cL)$ that is $\alpha$-\packable and $\lambda$-\safe. Then,
  \begin{OneLiners}
  \item[(i)] for all $X \subseteq [n]$, 
    the set system $(X, \cL|_X)$ is $\alpha$-\packable and $\lambda$-\safe, and 
  \item[(ii)]
    given a partition $X_1, \ldots, X_s$ of $[n]$, and set systems
    $(X_1, \cL_1 ), \ldots, (X_s, \cL_s )$, where
    $\cL_i  = \cL|_{X_i}$ for all $i$, the disjoint union of
    these systems is also $\alpha$-\packable.
  \end{OneLiners}
\end{restatable}
\begin{proof}
For the first statement, consider any $X\sse [n]$ and let $\cL|_X = \{L'_i\}_{i=1}^m$. The $\lambda$-\safe property follows  by using the same sets $M$ and $R_i$s for each $D\sse [m]$; note that  $L'_i=L_i\cap X$ for all $i\in [m]$. 
To see the $\alpha$-\packable property, consider any rewards $r_j$ and sizes $s_j \ge 0$ for elements $j\in X$, and threshold $\theta$. We extend these  rewards and sizes to the entire set $[n]$ by setting $r_j=s_j=0$ for all $j\in [n]\setminus X$. We now use the fact that the original set-system is $\alpha$-\packable.  Let $y\in [0,1]^n$ denote an  LP solution to \eqref{eq:det_LP}. Because $r_j=0$ for all $j\not\in X$, we can set    $y_j=0$ for all $j\not\in X$, without changing the objective.  Now, the rounded integer solution $\widehat{y}$ obtains at least a $1/\alpha$ fraction of the LP reward. Moreover, $\widehat{y}$ only selects elements in $X$ as the support of $\widehat{y}$ is contained in the support of $y$, which is contained in $X$. 

For the second statement, note that the LP constraint matrix in \eqref{eq:det_LP} for such a set-system is block-diagonal. Indeed, because of the disjoint union, constraints corresponding to resources in $\cL_h$ only involve variables corresponding to $X_h$, for all $h=1,\cdots s$. Let  $y^{(h)}$ denote the restriction of the LP solution $y$ to elements $X_h$, for each $h$. Then, using the $\alpha$-\packable property on $ (X_h, \cL_h )$, we obtain  an integral solution  $\widehat{y^{(h)}}$ that has at least a $1/\alpha$ fraction of the   reward from $y^{(h)}$. Combining the integer solutions $\widehat{y^{(h)}}$ over all $h=1,\cdots s$ proves the $\alpha$-\packable property for the disjoint union. 
  \end{proof}

We consider the \makespan problem for settings where the set system
$([n], \{L_i\}_{i \in [m]})$ is $\alpha$-\packable and $\lambda$-\safe for some small
parameters $\alpha$ and $\lambda$. We show in \S\ref{sec:apps} that the
families discussed above satisfy these properties.
Our main result is the following:
\begin{theorem}
\label{thm:main}
For any instance of \makespan where the corresponding set system
$([n], \{L_i\}_{i \in [m]})$ is $\alpha$-\packable and $\lambda$-\safe, there is an $O(\alpha \lambda \cdot \log \log m)$-approximation algorithm. 
\end{theorem}

\subsection{Effective Size and Random Variables}
\label{sec:effect-size-rand}

In all the arguments that follow, imagine that we have scaled the
instance so that the optimal expected
makespan is between $\frac12$ and $1$. 
It is  useful to split each random variable $X_j$ into two parts: 
\begin{OneLiners}
\item the {\em truncated } random variable  $X'_{j} := X_{j} \cdot \mathbf{I}_{(X_{j}\le 1)}$,  and
\item the {\em exceptional } random variable $X''_{j} := X_{j} \cdot \mathbf{I}_{(X_{j}> 1)}$.
\end{OneLiners}

These two kinds of random variables behave very differently with
respect to the expected makespan. Indeed, the expectation is a good
measure of the load due to   exceptional r.v.s, whereas one needs a
more nuanced notion for truncated r.v.s (as we discuss below). 
 The following result was shown in~\cite{KRT}:
\begin{lemma}[Exceptional Items Lower Bound]
  \label{lem:KRT-max-sum}
  Let $X_1'',X_2'',\dots,X_t''$ be non-negative discrete random
  variables each taking value zero or at least $L$. If
  $\sum_j \E[X_j'']\ge L$ then $\E[\max_j X_j'']\ge L/2$.
\end{lemma}

We now consider the trickier case of truncated random variables $X_j'$.
 We want to find a
deterministic quantity that is a good surrogate for each random
variable, and then use this deterministic surrogate instead of the
actual random variable.  
For stochastic load
balancing, a useful surrogate is the \emph{effective size}, which is
based on the logarithm of the (exponential) moment generating
function (also known as the cumulant generating function)~\cite{Hui,Kelly-notes,ElwaM,gupta2018stochastic}.

\begin{definition}[Effective Size]
	\label{defn:logmgf}
	For any r.v.\ $X$ and integer $k \geq 2$, define
	\begin{gather}
		\beta_k(X) \,\, :=\,\, \frac{1}{\log k} \cdot \log \E\Big[
		e^{(\log k) \cdot X}\Big].
	\end{gather}
	Also define $\beta_1(X) := \E[X]$.
\end{definition}
To see the intuition for the effective size, consider a set of
independent r.v.s $Y_1, \ldots, Y_k$ all assigned to the same
resource.  The following lemma, whose proof is very reminiscent of the
standard Chernoff bound (see \cite{Hui}), says that
the load is not much higher than the expectation.
\begin{restatable}[Effective Size: Upper Bound]{lemma}{Chern}
  \label{lem:chern}
  For indep.\ r.v.s $Y_1, \ldots, Y_n$, if $\sum_i \beta_k(Y_i) \leq
  b$ then $\P[ \sum_i Y_i \geq c ] \leq \frac{1}{k^{c-b}}$.
\end{restatable}
The usefulness of the effective size comes from a partial
converse~\cite{KRT}:
\begin{lemma}[Effective Size: Lower Bound]
\label{lem:krt-LB} Let $X_1,X_2,\cdots X_n$ be independent $[0,1]$ valued
	r.v.s, and $\{\widetilde{L_i}\}_{i=1}^m$  a partition of
	$[n]$.  
	If $\sum_{j=1}^n \beta_{m}(X_j)\ge 17m$ then $$\E\bigg[\max_{i=1}^m
	\sum_{j\in \widetilde{L_i}} X_j\bigg]=\Omega(1).$$
\end{lemma}

%%%%%%%%%%%%%%%%%%%%%%%%%%%%%%%%%%%%%%%%%
\section{The General Framework}
\label{sec:general}

In this section we prove Theorem~\ref{thm:main}: given a set system that is
$\alpha$-\packable and $\lambda$-\safe, we show an
$O(\alpha \lambda \log\log m)$-approximation algorithm. The idea is to
write a suitable LP relaxation for the problem (using the effective
sizes as deterministic surrogates for the stochastic tasks), to solve
this exponentially-sized LP, and then to round the solution. The
novelty of the solution is both in the LP itself, and in the rounding,
which is based on a delicate decomposition of the instance into
$O(\log \log m)$ many deterministic sub-instances.

In order to obtain an $O(\rho)$-approximation algorithm for \smm, 
it suffices to find a polynomial algorithm that does one of the following:
\begin{itemize}
    \item find a   solution of objective at most $\rho$, or
    \item prove that the optimal \smm value is more than $1$.
\end{itemize}
This follows from standard scaling ideas: see Appendix~\ref{app:scaling}. 
Henceforth, we will assume that the optimal value is at most $1$, and provide an algorithm that finds a solution with small expected makespan.

\subsection{The LP Relaxation}
\label{sec:lp-relaxation}

Consider an instance $\cI$ of \smm given by a set of $n$ tasks and $m$
resources, with sets $U_j$ and $L_i$ as described in
\S\ref{sec:prdef}. 
We now provide an LP relaxation which is feasible if the
optimal makespan is at most one. We use properties of truncated and
exceptional random variables; recall the definitions of these r.v.s from
\S\ref{sec:effect-size-rand}.

\begin{lemma}
	\label{lem:UF-lb}  
	Consider any feasible solution to $\cI$ that selects  a subset $S\sse [n]$ of tasks. If the expected maximum load $	\E\left[ \max_{i = 1}^m \sum_{j \in L_i\cap S} X_{j} \right]\le 1$, then
	\begin{align}
	\sum_{j\in S} \E[X''_{j}] &\le 2, \qquad \mbox{ and
	} \label{eq:UF-lb-1} \\
	\sum_{j\in  L(K)\cap S}  
	\, \beta_{k}(X'_{j}) \,\, &\le \,\, b\cdot k,\mbox{ for all }K\subseteq [m], \quad \mbox{where
	}k=|K|, \label{eq:UF-lb-2}
	\end{align}
   for $b$ being a large enough but fixed constant. 
\end{lemma}
\begin{proof}
	The first inequality~\eqref{eq:UF-lb-1} follows
	 from Lemma~\ref{lem:KRT-max-sum} applied to
	$\{X''_{j}\, :\, j\in S\}$ and $L=1$. 
	
	For the second inequality~\eqref{eq:UF-lb-2}, consider any subset
	$K\sse[m]$ of the resources. 
	 Let 	$\widetilde{L_i}\sse L_i$ for $i\in K$ be such that $\{\widetilde{L_i} \}_{i\in K}$  forms a partition of $\bigcup_{i\in K} (L_i \cap S)=L(K)\cap S$. Then,  we apply  
	Lemma~\ref{lem:krt-LB} to the resources in $K$ and the truncated 
	random variables $\{X'_{j} : j\in \bigcup_{i\in K} \widetilde{L_i} \}$. Because, $\E[\max_{i\in K} \sum_{j\in \widetilde{L_i}} X'_j]\le\E[\max_{i\in K} \sum_{j\in L_i\cap S} X'_j]\le 1$,  the contrapositive of Lemma~\ref{lem:krt-LB}  implies  $\sum_{j\in \bigcup_{i\in K}\widetilde{L_i}} \beta_k(X'_j)\le b\cdot k$, where $b=O(1)$ is a fixed constant. Inequality \eqref{eq:UF-lb-2} now follows from   $\bigcup_{i\in K}\widetilde{L_i}=L(K)\cap S$.
  \end{proof}

Lemma~\ref{lem:UF-lb} allows us to write the following feasibility linear
programming relaxation for \makespan (assuming the optimal value is
at most 1). For every task $j$, we have a binary variable $y_j$ corresponding to selecting $j$. 
\begin{alignat}{2}
\sum_{j=1}^n y_{j} &\ge t & & \label{eq:UFLP:assign}\\ 
\sum_{j=1}^n  \E[X_j''] \cdot y_{j}  &\le
2 & & \label{eq:UFLP:exceptn} \\ 
\sum_{j\in L(K)}  \beta_{k}(X'_{j}) \cdot y_{j} & \le b\cdot k &
 \qquad &\forall K\sse [m] \mbox{ with }|K|=k, \,\, \forall k=1,2,\cdots m,\label{eq:UFLP:subset}\\ 
0\le y_{j}  &\le 1&&\forall j\in[n]. \label{eq:UFLP:nonneg}
\end{alignat}

In the above LP, $b \geq 1$ denotes the universal constant multiplying $k$ in the
right-hand-side of~\eqref{eq:UF-lb-2}. Note that the effective sizes $\beta_k(X_j)$ can be computed in polynomial time because each $X_j$ has polynomial support-size. 
Despite having an exponential number of constraints,  this linear program can be solved
approximately in polynomial time. This relies on   the ellipsoid algorithm with an approximate separation oracle, see e.g.~\cite{CarrV02}.

\begin{restatable}[Solving the LP]{theorem}{LPSolve}
  \label{thm:LP}
There is a polynomial time algorithm which given an instance $\cI$ of \makespan outputs one of the following:
\begin{OneLiners}
\item a solution $y\in \R^n$ to LP
  \eqref{eq:UFLP:assign}--\eqref{eq:UFLP:nonneg}, except that the right-hand-side
  of~\eqref{eq:UFLP:subset} is  replaced by $\frac{e}{e-1}bk$, or
    \item a certificate  that LP \eqref{eq:UFLP:assign}--\eqref{eq:UFLP:nonneg} is infeasible.
\end{OneLiners}
\end{restatable}

\begin{proof} 
Our algorithm aims to satisfy the constraints \eqref{eq:UFLP:subset}, but will only achieve the following slightly weaker constraint:    \begin{equation}
    \sum_{j\in L(K)}  \beta_{k}(X'_{j}) \cdot y_{j}   \le \frac{e}{e-1} b\cdot k,\quad \forall K\sse [m] \mbox{ with }|K|=k, \,\, \forall k\in [m],\label{eq:UFLP:subset'}
    \end{equation}

We use the ellipsoid algorithm to find a feasible solution to the above LP. Given $y\in \R^n$ the separation oracle needs to check if constraint \eqref{eq:UFLP:subset} is satisfied (the other constraints are easy to check). To this end, we  use  the maximum-coverage problem. Given $n$ elements with non-negative weights $\{w_j\}_{j=1}^n$, a collection $\{S_i\sse [n]\}_{i=1}^m$  of subsets  and bound $k$, the goal  is to select $k$ subsets $T_1,\cdots T_k$ that maximize the total weight $\sum_{j\in \cup_{i=1}^k T_i} w_j$ of covered elements. There is an $\frac{e}{e-1}\approx 1.58$ approximation algorithm for maximum-coverage~\cite{cornuejols1977exceptional}.

For each $k$, $1\le k\le m$, we consider  an instance $\I_k$  of the  maximum-coverage problem with $m$ sets $\{L_i\}_{i=1}^m$ and weights $w_j=\beta_k(X'_j)\cdot y_j$ on each task $j\in[n]$. Note that checking \eqref{eq:UFLP:subset} for subsets $K$ of size $k$ is equivalent to checking if the optimal value of $\I_k$ is at most $bk$. Let $A_k \sse [m]$ denote the approximate solution to $\I_k$ that we obtain for each $k$ by using the algorithm from~\cite{cornuejols1977exceptional}. Then we have the following cases: 
\begin{itemize}
    \item For some $k$, the value  $\sum_{j\in L(A_k)}  \beta_{k}(X'_{j}) \cdot y_{j}$ is more than $bk$. Then, this is a violated constraint, which can be added to the ellipsoid algorithm.
    \item For each  $k$, the value $\sum_{j\in L(A_k)}  \beta_{k}(X'_{j}) \cdot y_{j}$ is at most $bk$. Then it follows that, for each $k$,  the optimal value of $\I_k$ is at most $\frac{e}{e-1} b  k$. This implies that  constraint~\eqref{eq:UFLP:subset'} is satisfied. 
\end{itemize}
This proves the desired result. 
   \end{proof}

In the rest of this section, we assume we have a feasible solution $y$
to~\eqref{eq:UFLP:assign}--\eqref{eq:UFLP:nonneg}, and ignore the fact
that we only satisfy~\eqref{eq:UFLP:subset} up to a factor of
$\frac{e}{e-1}$, since it only affects the approximation ratio by a
constant factor.

Note that we only use effective sizes $\beta_k$ of {\em truncated} r.v.s, so we have $0\le \beta_k(X'_j)\le 1$ for all $k\in [m]$ and $j\in [n]$. 
Moreover, we make the following assumption (without loss of generality)  on the exceptional r.v.s. 
\begin{assumption}\label{asmp:excep}
We have $\E[X''_j] \le 2$  for every task $j\in [n]$.
\end{assumption}
Indeed, we can simply 
drop all tasks $j$ with $\E[X''_j]>2$ as such a task would never be part of an optimal solution- by \eqref{eq:UF-lb-1}.

\subsection{Overview of Analysis}

Let us give some intuition behind the factor of $O(\log \log m)$ that
arises in the approximation ratio. To keep things concrete, consider
the special case of intervals on a line: each task is an interval, and
each of the $m$ resources is a point on the line. Each task (interval)
loads all the resources (points) that lie within the interval. For
simplicity, consider the special case where we have an {\em integral}
solution $y$ to the LP relaxation
\eqref{eq:UFLP:assign}--\eqref{eq:UFLP:nonneg}, and therefore there is
no need to perform any rounding. (Our analysis loses a $\log\log m$
factor even in this special integral case.)  Let $T$ denote the set of
intervals for which $y_j = 1$. We would like to argue that the
expected makespan due to selecting set $T$ is $O(\log \log m)$.

To this end, we partition the points into $O(\log\log m)$ groups such
that (roughly speaking) the expected makespan due to each group is
$O(1)$.  We maintain a variable $k$ which is initialized to $2$, and a
set $J$ of remaining intervals (initially equal to $T$).  Consider a
greedy procedure to build an ordering $i_1, i_2, \ldots, i_m$ on the
points as follows.  Given $i_1, \ldots, i_r$,
 define $i_{r+1}$ to be any point $i$ for which
$\sum_{j \in J \cap L_i } \beta_k(X_j') > b$, and remove all
intervals containing this point $i$ from $J$.  If there is no such point then we update
$k$ to $k^2$, and continue. (If $k$ exceeds $m$, order the remaining
points arbitrarily.) Observe that $k$ takes values which are of the
form $k_\ell := 2^{2^\ell}$ for non-negative integers $\ell$. We refer
to the index $\ell$ as the ``class''. For each class $\ell$, let
$D_\ell$ denote the set of points in the above sequence that were
added when $k$ was equal to
$k_\ell$.  
Note that $|D_\ell| \leq k_\ell$ -- this follows directly
from~\eqref{eq:UFLP:subset}. Indeed, if $|D_\ell| > k_\ell$, then by
choosing any $k_\ell$ points from $D_\ell$,
constraint~\eqref{eq:UFLP:subset} (with $k= k_\ell$) would not be
satisfied. For each class $\ell$, let $J_\ell\sse T$ denote those
intervals that were removed from $J$ when $k=k_\ell$.

We now argue about the makespan of $D_\ell$ (i.e., the class-$\ell$ points) due to $J_\ell$  (the class-$\ell$ intervals). 
 We first observe that for any point $i \in D_\ell$,  $\sum_{j \in
  J_\ell \cap L_i} \beta_{k_{\ell-1}}(X_j')$ is at most $b$. Indeed,
if not, this point $i$ would have been added to $D_{\ell-1}$ instead. 
 We now apply Lemma~\ref{lem:chern}: for any point $i\in D_\ell$, the probability that intervals  $J_\ell$   load   $i$  to  more than $b+4$ is at most $k_{\ell-1}^{-4}=  k_{\ell}^{-2}$. Then, by a union bound over all points in $D_\ell$, 
the probability that intervals in $J_\ell$ load {\em any} point in
$D_\ell$ to more than $b+4$ is at most $|D_\ell|\cdot k_{\ell}^{-2}\le
\frac{1}{k_\ell}$.
 With some additional work, we can also show that the expected makespan
of $D_\ell$ (due to intervals $J_\ell$) is a constant. These arguments
are formalized (for the general setting) in
Lemmas~\ref{obs:beta_class} and~\ref{obs:UF_class_size}.

However, for any particular point $i \in D_\ell$, we also need to
worry about intervals which are in $(T \setminus J_\ell)\cap
L_i$. These intervals must belong to previous classes, by the
construction of the ordering. For each class $\ell' < \ell$, consider
the two points in $D_{\ell'}$ closest to $i$ on either side: this
gives us at most $2\log \log m$ such ``representative'' points. Any
interval in $(T \setminus J_\ell)\cap L_i$ would load at least one of
these representatives. Hence, we can bound the load from these
intervals by the total load on the representatives, which is
$O(\log \log m)$ in expectation. This is formalized by the
$\lambda$-\safe property and Lemma~\ref{lem:class-makespan}.

In this overview, we omitted the issue of rounding the LP
solution. This is handled by classifying tasks as being large/small based on their
$y_j$ value (see Lemma~\ref{lem:large-tasks}) and using the
$\alpha$-\packable property (see Lemmas \ref{lem:det-feas} and
\ref{obs:UF_resource-tail}).

\subsection{The Deterministic Subproblem}
We actually need a slight generalization of the reward-maximization problem
mentioned in~\eqref{eq:det_LP}, which we call the \detc problem.  An
instance $\cI$ of the \detc problem consists of a set system $([n],
\cS)$, with a size $s_j$ and cost $c_j$ for each element $j \in [n]$. It
also has parameters $\theta \geq \max_j s_j$ and $\psi \geq
\max_j c_j$. The goal is to find a maximum cardinality subset $V$ of
$[n]$ such that each set in $\cS$ is ``loaded'' to at most $\theta$,
and the total cost of $V$ is at most $\psi$. We use  the following  LP relaxation:
\begin{alignat}{2}
\max & \sum_{j=1}^n y_{j} &    \label{eq:det-LP-mod}\\ 
s.t. \,\,&  \sum_{j\in S} s_j\cdot  y_j  \le \theta, & \quad \forall S \in \cS \notag \\
&\sum_{j \in [n]}  c_j\cdot y_j \le \psi, &  \notag \\
& 0\le y_j \le 1, & \forall j \in [n]\notag 
\end{alignat}

\noindent
The following result, 
which motivates the $\alpha$-packable property,
shows
that the $\alpha$-packable property for a set system implies an
$O(\alpha)$-approximation for the \detc problem. 

\begin{restatable}{theorem}{detsolve}
\label{thm:det-prob}
Suppose a set system satisfies the $\alpha$-packable property. Then there is
an $O(\alpha)$-approximation algorithm for   \detc
relative to the LP relaxation~\eqref{eq:det-LP-mod}. 
\end{restatable}

\begin{proof}
 Consider an instance $\I$ of \detc  consisting of a set system $([n],
 \cL)$, cost $c_j$ and size $s_j$ for each element $j \in [n]$, and
 parameters $\theta \geq \max_j s_j$ and $\psi \geq \max_j c_j$. Let
 $y$ be a solution to the LP~\eqref{eq:det-LP-mod},  with objective
 function value $T = \sum_j y_j$. We construct an instance $\I'$ of
 the reward-maximization problem with LP
 relaxation~\eqref{eq:det_LP}. The set system, sizes of elements and
 the parameter $\theta$ are as in $\I$. Furthermore, the reward $r_j$
 of an element $j$ is defined as:
 \begin{equation*}
   \textstyle r_j := \left( 1-\frac{T}{2\psi}c_j \right).
 \end{equation*}
 Since the set of constraints in~\eqref{eq:det_LP} is a subset of that
 in~\eqref{eq:det-LP-mod}, the solution $y$ is also a feasible solution to~\eqref{eq:det_LP} with objective function value equal to 
 $$ \sum_{j \in [n]} r_j y_j =\sum_{j \in [n]} \left(1-\frac{T}{2\psi}c_j \right) y_j \geq T - T/2 = T/2. $$
 The   inequality uses the fact that $\sum_{j=1}^n c_j y_j  \leq \psi$. Now the
 $\alpha$-\packable property implies that we can find a subset $S \sse [n]$ which is a feasible integral solution to~\eqref{eq:det_LP}, whose total reward $\sum_{j \in S} r_j \ge \frac{T}{2 \alpha}$. Since $r_j \leq 1$ for all $j$, it follows that $|S| \geq  \frac{T}{2 \alpha}$ as well. Moreover,  by definition of  $r_j$, we have that
 $\sum_{j \in S} r_j = |S| - \frac{T}{2\psi} c(S)\ge \frac{T}{2 \alpha} \ge 0$. Hence, 
\begin{equation}\label{eq:detc-soln}
|S| \geq \frac{c(S) }{2\psi} T.
\end{equation}

 If the total cost of the elements in $S$ is at most $\psi$, this is
 also a feasible solution to $\I$ with $|S|\ge \frac{T}{2\alpha}$.

 Below, we assume
 that $c(S) = \sum_{j\in S} c_j > \psi$.  Starting with a partition of $S$ into singletons, we repeatedly merge any two  parts whose total cost is at most $\psi$. Let  
 $S_0, \ldots, S_{u-1}$ denote the parts at the end of this process. As each element has cost at most $\psi$ and we only merge parts when their total cost is at most $\psi$, it follows that the cost $c(S_k)$ of each part $S_k$  is at most $\psi$. Moreover, the total cost of {\em any} pair of parts is more than $\psi$. This implies that $c(S_k)+c(S_{k+1})>\psi$ for each $k=0,\ldots u-1$ (the indices are modulo $u$). Adding these $u$ inequalities, we have $2\sum_{k=0}^{u-1} c(S_k) > u\cdot \psi$ which implies $u< \frac{2}{\psi} \sum_{k=0}^{u-1} c(S_k) = \frac{2}{\psi} c(S)$.  
  Let $S^*$ be the maximum cardinality set among $\{S_k\}_{k=0}^{u-1}$. Note that
 $|S^*|\ge \frac1u |S| > \frac{\psi}{2c(S)}|S|$. Using \eqref{eq:detc-soln} we obtain  $|S^*| \geq \frac{T}4$.  
 
 So in either case, we
 are guaranteed an $\bar{\alpha} = \max\{2\alpha,4\} = O(\alpha)$
 approximation for \detc relative to the LP. This completes the proof.
  \end{proof}

\subsection{Rounding a Feasible LP Solution}
\label{sec:rounding}

We first give some intuition about the \makespan rounding algorithm. It involves formulating $O(\log\log m)$ many almost-disjoint
instances of the deterministic reward-maximization
problem~\eqref{eq:det_LP} used in the definition of
$\alpha$-packability. The key aspect of each  deterministic
instance is the definition of the sizes $s_j$: for the $\ell^{th}$
instance we use effective sizes $\beta_k(X_j')$ with parameter
$k=\smash{2^{2^\ell}}$. We use the $\lambda$-safety property to
 construct these deterministic instances and the $\alpha$-packable
property to solve them. Finally, we show that the  expected makespan
induced by the selected tasks is at most $O(\al\lambda)$ from each deterministic instance, which leads to an overall $O(\al\lambda\log\log m)$-approximation ratio. The
procedure is described formally in Algorithm~\ref{alg:round}. 

\begin{algorithm}[htb]\label{alg:round}  
	\caption{Rounding Algorithm}  
	\LinesNumbered
	\SetKwInOut{Input}{Input}\SetKwInOut{Output}{Output}
	\SetKwFor{Do}{Do}{\string:}{end}
	\Input{A fractional solution $y$ to \eqref{eq:UFLP:assign}--\eqref{eq:UFLP:nonneg}}
	\Output{A subset of tasks.}
	Initialize remaining tasks $J\leftarrow [n]$\;
	\For {$\ell = 0,1,\dots, \log\log m$} 
	{Set $k\gets 2^{2^\ell}$\; \label{l:for1}
	Initialize class-$\ell$ resources $D_\ell\leftarrow\emptyset$\;
		\While{there is a resource $i \in [m] \,:\, \sum_{j \in L_i\cap J}\beta_{k^2}(X_j')\cdot y_j\,>\,2b$}{
		\label{l:ksq}	 update $D_\ell\leftarrow D_\ell\bigcup\{i\}$\; 
		Set $\widetilde{L_i} \leftarrow J\cap L_i$ and $J \leftarrow J\setminus \widetilde{L_i}$\label{l:Ltilde}\; 
			}
		Define the class-$\ell$ tasks $\J_\ell\leftarrow \bigcup_{i\in D_\ell}\widetilde{L_i}$ 	\label{l:jell}\;
                Use $\lambda$-safety on the set system $(J_\ell, \{L_i
                \cap J_\ell\}_{i \in [m]})$ to get $M_\ell := \ext(D_\ell)$ 		\label{l:fore}\;
			} 
		$\rho \leftarrow 1+\log\log m$\;
		Define class-$\rho$ tasks $\J_{\rho} = J$  and class-$\rho$ resources $M_\rho := D_\rho=[m]\setminus \left(\bigcup_{\ell =0}^{\rho-1} D_\ell\right)$ \; 
		Define an   instance \dI of \detc as follows: the set system is the  disjoint union of the set systems $(J_\ell, M_\ell)$ for $\ell = 0, \ldots, \rho$. The other parameters are as follows:   
		$$\mbox{Sizes }s_j = \beta_{2^{2^\ell}}(X'_j)\,\, \mbox{ for each }j\in \J_\ell \mbox{ and } 0\le \ell \le \rho, \,\, \mbox{ bound } \theta=2\bar{\alpha}b, $$  
		$$\mbox{Costs }c_j = \E[ X''_j]\,\, \mbox{ for each }j\in [n],  \,\, \mbox{ bound } \psi=2\bar{\alpha}, $$     where $\bar{\alpha}$ is the approximation ratio from Theorem~\ref{thm:det-prob}  		\label{l:detcost} \;
		Let $N_H = \{j \in [n] \,:\, y_j> 1/\bar{\alpha}\}$ 		\label{l:nh} \;
		Let $\bar{y}_j = \bar{\alpha}\cdot y_j$ for $j\in [n]\setminus N_H$ and $\bar{y}_j=0$ otherwise \label{l:bary}
 \;
        Round $\bar{y}$ (as a feasible solution
        to~\eqref{eq:det-LP-mod})	 using  Theorem \ref{thm:det-prob} to
        obtain $N_L$  \label{l:callLP}  \; 
        Output  $N_H\bigcup N_L$.
\end{algorithm}

The
algorithm proceeds in $\log \log m$ iterations of the {\bf for} loop in
Lines~\ref{l:for1}--\ref{l:fore}. The set $J$ denotes the remaining tasks at any point in the algorithm.  In each iteration $\ell$, we make use of effective sizes $\beta_k$ with parameter $k=2^{2^\ell}$ (see Line~\ref{l:for1}). In Line \ref{l:ksq}, we identify
resources $i$ which are fractionally loaded to more than $2b$, where
the load is measured in terms of $\beta_{k^2}(X_j')$ values and we only consider the remaining tasks $J$. The set
of such resources is grouped in the set $D_\ell$ (called the class-$\ell$ resources). We also define
the class-$\ell$ tasks $J_\ell$ to be all remaining tasks (in $J$) which can load the resources $D_\ell$. Ideally, we
would like to remove these resources and tasks, and iterate on the
remaining tasks and resources.  However, the problem is that tasks in
$J_\ell$ also load resources other than $D_\ell$, and so  
$(D_\ell, J_\ell)$ is not  independent of the rest of the
instance. This is where we use the $\lambda$-safe property: in Line~\ref{l:fore} we expand $D_\ell$ to a
larger set of resources $M_\ell:= \ext(D_\ell)$, which will be used to bound the load induced by  $J_\ell$ on resources outside $D_\ell$. We use $(J_\ell, M_\ell)$ to represent the set system corresponding to  class-$\ell$: note that each set is of the form $L_i \cap J_\ell$ for some $i \in M_\ell$.

Having partitioned the tasks into classes $J_1, \ldots, J_{\rho}$,  we consider
the disjoint union ${\cal D}$ of the set systems $(J_\ell, M_\ell),$ for
$\ell = 1, \ldots, \rho$. While the sets $D_\ell$ are disjoint,  the sets $M_\ell$ may not be
disjoint. For each resource appearing in multiple sets $M_\ell$, we 
make distinct copies in the combined set-system  ${\cal D}$. 
Then we set up an instance $\dI$ of \detc (in
Line~\ref{l:detcost}): the set system is ${\cal D}$, the disjoint union of
$(J_\ell, M_\ell),$ for $\ell = 1, \dots, \rho$. Every task
$j \in J_\ell$ has size $\beta_{2^{2^\ell}}(X_j')$ and cost
$\E[X_j''].$ The parameters $\theta$ and $\psi$ are as mentioned in
Line~\ref{l:detcost}. In Line~\ref{l:nh}, we include into our solution, all tasks ($N_H$) that have a large  LP value. Then, we define a scaled-up fractional solution ${\bar y}$ (in Line~\ref{l:bary}) supported  on all other tasks $[n]\setminus N_H$: we  will show later that this is feasible to the LP
relaxation~\eqref{eq:det-LP-mod} for $\dI$. Finally, we use 
Theorem~\ref{thm:det-prob} to round $\bar y$ to an integral solution
$N_L$ (in Line~\ref{l:callLP}) which is added to our solution.

\subsection{The Analysis}
\label{sec:analysis}
We now show that the expected makespan for the solution produced by
the rounding algorithm above is $O(\alpha \lambda \rho)$, where
$\rho=1+\log\log m$ is the number of classes. In particular, we show that the
expected makespan (taken over all resources) due to the selected tasks from each
class $\ell$ is $O(\alpha\lambda)$.

Our first lemma shows that the fractional load on every resource due to class-$\ell$ tasks (using effective size $\beta_{2^{2^\ell}}$) is at most a constant. 

\begin{lemma}
\label{obs:beta_class}
	For any  class $\ell$, $0\le \ell\le \rho$, and  resource $i\in [m]$, 
	$$\sum_{j \in \J_\ell \cap L_i}\beta_{r}(X_j')\cdot y_j\le 2b,\quad \mbox{ where }r=2^{2^\ell}. $$
	\end{lemma}
\begin{proof}
   If $\ell=0$, 
   we have $r=2$. Using the LP constraint~\eqref{eq:UFLP:subset} for a subset $\{i,i'\}$ of size two containing the resource $i$, we have: 
 	$$\sum_{j \in J_\ell \cap L_i} \beta_{2}(X_j')\cdot y_j \le\sum_{j \in L_i} \beta_{2}(X_j')\cdot y_j \le \sum_{j \in L(\{i,i'\})} \beta_{2}(X_j')\cdot y_j \le 2b,$$
   which implies the desired result. 
   
   So assume $\ell \geq 1$. Let $J$ denote the set of remaining tasks at the end of iteration $\ell-1$, i.e., $J= \bigcup_{\ell' \geq \ell} J_{\ell'}$.
   The terminating condition in Line \ref{l:ksq} (for iteration $\ell-1$) implies that 
   $$ \sum_{j \in J\cap L_i} \beta_{r}(X_j')\cdot y_j \leq 2b, \,\, \mbox{for all }i\in [m],$$
   which implies the lemma. 
  \end{proof}

Next, we bound the sizes of sets $D_\ell$ and $M_\ell$ as functions of $\ell$.
\begin{lemma}\label{obs:UF_class_size}
	For any $\ell$, $0 \leq \ell \leq \rho,$   $|D_\ell|\le k^2$, where $k=2^{2^{\ell}}$. So $|M_\ell|\le k^p$ for some  constant $p$.
\end{lemma}
\begin{proof} The lemma is trivial for the last class $\ell=\rho$ as $k\ge m$ in this case. Now consider any class $\ell<\rho$. Using the condition in Line~\ref{l:ksq}, we have:
\begin{equation}\label{eq:load-frac}
\sum_{j \in \widetilde{L_i}}\beta_{k^2}(X_j')\cdot y_j> 2b, \quad 	\forall i \in D_\ell,
\end{equation}
where $\widetilde{L_i}$ is as defined in Line~\ref{l:Ltilde}.
	Note that the subsets $\{\widetilde{L_i} : i\in D_\ell\}$ are disjoint as the set $J$ gets  updated (in Line \ref{l:Ltilde}) after adding each $i\in D_\ell$. Suppose, for the sake of contradiction, that $|D_\ell|>k^2$. Letting $K\subseteq D_\ell$ be any set of size $k^2$, we have:
		$$2b \cdot k^2  < \sum_{i\in K}\sum_{j \in \widetilde{L_i}}\beta_{k^2}(X_j')\cdot y_j \le \sum_{j\in L(K)} \beta_{k^2}(X_j')\cdot y_j \le b|K| = b\cdot k^2,$$ 
		which is a contradiction. Above, the first inequality uses \eqref{eq:load-frac} and $K\sse D_\ell$, and the last inequality uses the LP constraint \eqref{eq:UFLP:subset} on subset $K$. This proves the first part  of the lemma.
	Finally, the $\lambda$-\safe property implies that  $|M_\ell|$ is polynomially bounded by $|D_\ell|$, which proves the second part. 
  \end{proof}

We now show that the fractional solution $\bar y$ from Line \ref{l:bary} is feasible to the LP relaxation for \detc given in~\eqref{eq:det-LP-mod}.
\begin{lemma}\label{lem:det-feas}
    The fractional solution $\bar{y}$ is feasible for the LP
    relaxation~\eqref{eq:det-LP-mod} corresponding to  the \detc instance
    \dI. Moreover, we have $\max_j s_j\le \theta $ and $\max_j c_j\le \psi $ in instance
    \dI.
\end{lemma}
\begin{proof}
    Note that $0\le \bar{y}\le 1$ by construction. Since the sets $J_\ell$ partition $[n]$,  $$\sum_{\ell=0}^\rho \sum_{j \in J_\ell}  c_j \cdot \bar{y}_j
    = \sum_{j \in [n]} c_j\cdot \bar{y}_j \le \bar{\alpha} \sum_j c_j\cdot y_j \le 2 \bar{\alpha} = \psi$$ where the last inequality follows from the feasibility of  constraint~\eqref{eq:UFLP:exceptn}.
    
    To verify the size constraint for each resource $i$ in the
    disjoint union of $M_\ell$ for $\ell = 0, \ldots, \rho$, consider any such class $\ell$ 
     and $i\in M_\ell$. The size constraint for $i$ is:
    \begin{equation}\label{eq:det-LP-size-constr}
        \sum_{j\in \J_\ell\cap L_i} \beta_k(X_j')\cdot \bar{y}_j \le \theta = 2\bar{\alpha}b,
    \end{equation}
    where  $k=2^{2^\ell}$. Since $\bar{y}\le \bar{\alpha} \cdot y,$ this follows directly from Lemma~\ref{obs:beta_class}.
    
    Finally, since the truncated sizes $X'_j$ lie in $[0,1]$, so do
    their effective sizes. Hence $s_j \le 1 \leq
    \theta$ for all $j\in[n]$. Moreover, by Assumption~\ref{asmp:excep}  we have  $c_j = \E[X_j''] \leq 2 \leq
    \psi$ for all $j\in [n]$.  
      \end{proof}

Based on this lemma, we can indeed apply  Theorem~\ref{thm:det-prob} to round $\bar{y}$ into an integer
    solution (as done in Line \ref{l:callLP}). We now analyze our solution $N_H\bigcup N_L$. Recall that $N_H$ consists of all tasks $j$ with $y_j>1/\bar{\alpha}$ and $N_L$ is the rounded solution obtained from $\bar{y}$. 
    \begin{lemma}\label{lem:large-tasks}
The solution obtained in Algorithm~\ref{alg:round} has  $|N_H|+|N_L|\ge t$.
\end{lemma}
\begin{proof}
    Note that by the feasibility of the constraint~\eqref{eq:UFLP:assign},  $\sum_{j\in [n]\setminus N_H} y_j \ge t-|N_H|$. Further, $\bar{y}_j = \bar{\alpha}\cdot y_j \in [0,1]$ for all tasks $j \in [n]\setminus N_H$.  
     Therefore,
    $$|N_L| \ge \frac{1}{\bar{\alpha}}\sum_{j\in [n]\setminus N_H} \bar{y}_j = \sum_{j\in [n]\setminus N_H} y_j \ge t-|N_H|,$$
    which completes the proof. 
  \end{proof}

\def\ld{\mathsf{Load}^{(\ell)}}
We now bound the expected makespan of our solution $N := N_H\bigcup N_L$. 
We will focus on a particular class $\ell\le \rho$ and show that the expected makespan due to tasks in $N\cap \J_\ell$ is small. Recall that $k=2^{2^\ell}$. For sake of brevity, let $N_\ell := N\cap J_\ell$ be the selected class-$\ell$ tasks, and let $\ld_i := \sum_{j\in N_\ell\cap  L_i} X'_{j}$ denote the load on any resource $i\in [m]$ due to the selected class-$\ell$ tasks.  The following lemma can be viewed as the ``rounded'' version of Lemma~\ref{obs:beta_class}. 

\begin{lemma}\label{obs:UF_resource-tail}
	For any class $\ell\le\rho$ and resource $i\in M_\ell$, 
$$\sum_{j\in N_\ell\cap  L_i} \beta_k(X'_j) \le 4\bar{\alpha} b, \quad \mbox{where }k=2^{2^\ell}.$$
\end{lemma}
\begin{proof}
Since $N_\ell \cap L_i = (N_H \cap J_\ell \cap L_i) \bigcup (N_L \cap J_\ell \cap L_i),$ we bound the left-hand-side above in two parts. 
By Lemma~\ref{obs:beta_class}, the solution $y$ has $\sum_{j \in \J_\ell \cap L_i}\beta_{k}(X_j')\cdot y_j\le 2b$. As each task $j\in N_H$ has $y_j>1/\bar{\alpha}$,
$$ \sum_{j\in N_H \cap J_\ell \cap  L_i} \beta_k(X'_j) \le 2\bar{\alpha} b.$$

Since $N_L$ is a feasible integral solution to~\eqref{eq:det-LP-mod}, the size constraint for $i \in M_\ell$ implies that 
 $$\sum_{j\in N_L\cap J_\ell \cap L_i}  \beta_k(X'_j)  = \sum_{j\in N_L\cap J_\ell \cap L_i}  s_j \le \theta = 2\bar{\alpha}b.$$ 
Combining the two bounds above, we obtain the claim.
  \end{proof}

We are now ready to bound the makespan due to the truncated part of the random variables.

\begin{lemma}\label{lem:class-makespan}
For any class $\ell\le\rho$, we have $\E\left[ \max_{i\in M_\ell} \ld_i\right] \le 4\bar{\alpha} b +O(1)$  and therefore, 
$\E\left[ \max_{i=1}^m \ld_i\right] \le 4\lambda \bar{\alpha} b +O(\lambda) = O(\alpha\lambda)$.  
\end{lemma}
\begin{proof} 
Consider a resource $i \in M_\ell$. 
 Lemma~\ref{obs:UF_resource-tail} and  Lemma~\ref{lem:chern} imply that for any $\gamma > 0$, 
	$$\P\left[ \ld_i > 4\bar{\alpha}b +\gamma \right] = \P\left[  \sum_{j\in  N_\ell\cap  L_i} X'_{j} > 4\bar{\alpha}b +\gamma \right] \le k^{-\gamma}.$$
By a union bound, we get
$$\P\left[ \max_{i\in M_\ell} \ld_i > 4\bar{\alpha}b +\gamma \right]  \le |M_\ell|\cdot  k^{-\gamma} \le k^{p-\gamma},\qquad \mbox{ for all }\gamma\ge 0,$$
where $p$ is the constant from Lemma~\ref{obs:UF_class_size}. So the expectation
\begin{align*}
\E\left[\max_{i \in M_\ell}\ld_i\right]
& = \int_{\theta=0}^\infty \P\left[ \max_{i \in M_\ell} \ld_i > \theta \right] d\theta\\
&\le \,\, 4\bar{\alpha} b + p + 2 +\int_{\gamma=p+2}^\infty \P\left[ \max_{i \in M_\ell} \ld_i > 4\bar{\alpha} b +\gamma\right] \, d\gamma  \\
&\le \,\, 4\bar{\alpha} b + p + 2 +\int_{\gamma=p+2}^\infty k^{-\gamma+p}\, d\gamma  \,\, \le \,\,  4\bar{\alpha} b + p + 2 + \frac{1}{k(k-1)},
\end{align*}
which completes the proof of the first statement.

We now prove the second statement. Consider any class  $\ell<\rho$: by definition of $J_\ell$, we know that
 $\J_\ell\sse L(D_\ell)$. The $\lambda$-\safe property implies that for every resource $i\in [m]$ there is a subset $R_i\sse M_\ell$ with $|R_i|\le \lambda$ and $L_i\cap L(D_\ell)\sse L(R_i)$; using $\J_\ell\sse L(D_\ell)$  the latter property  implies  
 $L_i\cap J_\ell \sse L(R_i)\cap J_\ell$. Because $N_\ell\sse J_\ell$, we also have $L_i\cap N_\ell \sse L(R_i)\cap N_\ell$. 
Therefore, 
$$ \ld_i \le  \sum_{z \in R_i} \ld_z \le \lambda \,\max_{z\in M_\ell} \ld_z.$$
Taking expectation on both sides and using the first statement in the lemma, we obtain the desired result. 

Finally, for the last class $\ell=\rho$, note that any task in $\J_\rho$ loads only the resources in $D_\rho = M_\rho$. 
Therefore, $\max_{i=1}^m \ld_i = \max_{z\in M_\ell} \ld_z$. Taking expectation on both sides, we obtain the second statement. 
  \end{proof}

Using Lemma~\ref{lem:class-makespan}, we can bound the expected makespan due to all truncated random variables:
{\small \begin{equation}\label{eq:truc-makespan}
    \E\left[ \max_{i=1}^m \sum_{j\in N\cap L_i} X'_j \right] =     \E\left[ \max_{i=1}^m \sum_{\ell=0}^\rho \ld_i\right] \le      \sum_{\ell=0}^\rho  \E\left[ \max_{i=1}^m \ld_i\right] \le O(\alpha\lambda\rho).
\end{equation}}
The next lemma handles exceptional random variables.
\begin{lemma}\label{cl:excp-makespan}
    $\E\left[ \sum_{j\in N} X''_j \right] = \sum_{j\in N} c_j   \le 4 \bar{\alpha}$.  
\end{lemma}
\begin{proof}
    Feasibility of  constraint \eqref{eq:UFLP:exceptn} implies that  $\sum_{j=1}^n c_j\cdot y_j\le 2$.  As each task $j\in N_H$ has $y_j>1/\bar{\alpha}$, we have $\sum_{j\in N_H} c_j \le 2 \bar{\alpha}$. For tasks in $N_L$, the fact that $N_L$ is a feasible integral solution to~\eqref{eq:det-LP-mod} implies that 
     $\sum_{j\in N_L} c_j \le \psi = 2\bar{\alpha}$. This completes the proof. 
  \end{proof}

Finally, using \eqref{eq:truc-makespan} and Claim~\ref{cl:excp-makespan}, we have:
\begin{align*}
\E\left[ \max_{i=1}^m \sum_{j\in N\cap L_i} X_j \right] &=  \E\left[ \max_{i=1}^m \sum_{j\in N\cap L_i} (X'_j+X''_j) \right] \\
&\le   \E\left[ \max_{i=1}^m \sum_{j\in N\cap L_i} X'_j \right]  \,\,+\,\, \E\left[ \sum_{j\in N} X''_j \right] \le O(\alpha\lambda\rho).
\end{align*}
This completes the proof of Theorem~\ref{thm:main}. 

%%%%%%%%%%%%%%%%%%%%%%%%%%%%%%%%%%%%%%%%%%%%%%%

\section{Applications}
\label{sec:apps}

In this section, we show that several stochastic optimization problems
of interest  satisfy the two
assumptions of $\alpha$-packability and $\lambda$-safety for small
values of these parameters (typically $\alpha, \lambda = O(1)$ in
these problems). Hence \smm can be solved efficiently using our framework.

\subsection{Intervals on a Line}
\label{sec:intline}

We are given a path graph on $n$ vertices, which we call a line.  The
resources are the vertices in this line. Each task corresponds to an
interval in this line and loads all the vertices in the corresponding
interval. For each vertex $i$, $L_i$ denotes the subset of tasks
(i.e., intervals) which contain~$i$.

The $\alpha$-\packable property for this set system with $\alpha = O(1)$ follows from the result in~\cite{chakrabarti2007approximation}---indeed, the LP relaxation~\eqref{eq:det_LP} corresponds to the unsplittable flow problem where all vertices have uniform capacity $\theta$. We now show the $\lambda$-\safe property. 

\begin{lemma}\label{lem:UFline_set}
The above set system is 2-\safe.
\end{lemma}
\begin{proof}
Consider a subset $D$ of vertices. We define $M:= \ext(D)$ to be same as $D$. For a vertex $i$, let $l_i$ and $r_i$ denote the closest vertices in $M$ to the left and to the right of $i$ respectively (if $i \in M$, then both these vertices are same as $i$). Define $R_i$ as $\{l_i, r_i\}$. It remains to show that $L_i \cap L(D) \subseteq L(R_i)$. This is easy to see. Consider a task $j$ (represented by interval $I_j$) which belongs to $L_i \cap L(D)$. Then $I_j$ contains $i$ and a vertex from $D$. But then it must contain either $l_i$ or $r_i$. Therefore, $j$ belongs to $L(R_i)$ as well. 
   \end{proof}

Theorem~\ref{thm:main} now implies the following. 
\begin{corollary}\label{cor:UFline}
There is an $O(\log\log m)$-approximation algorithm for   \makespan where
the resources are represented by vertices on a line and tasks by intervals in this line. 
\end{corollary}

\subsection{Paths on a Tree}
We are given a tree $T=(V,E)$ on $|V|=m$ vertices, and a set of $n$ paths, $\{P_j\}_{j=1}^n$, in this tree. The resources correspond to vertices and the tasks correspond to paths. For a vertex $i \in [m]$, $L_i$ is  the set of paths which contain $i$. We first show the $\lambda$-\safe property. 

\begin{lemma}\label{lem:UFtree_set}
The set system $([n], \{L_i: i \in [m]\})$ is 2-safe. 
\end{lemma}
\begin{proof}
    Let $D$ be a subset of vertices. We define $M:= \ext(D)$ as follows: let $T'$ be the minimal sub-tree of $T$ which contains all the vertices in $D$. Note that all leaves of $T'$ must belong to $D$. 
    Then $M$ contains $D$ and all the vertices in $T'$ which have degree at least three (in the tree $T'$). It is easy to check that $|M| \leq 2|D|$. Fix a vertex $i \in V$. We need to define $R_i$ such that $L_i\cap L(D)\sse  L(R_i)$. 
	Let $v_i$ be the vertex in the sub-tree $T'$ that has the least distance to $i$ (if $i \in T'$, then $v_i$ is same as $i$). Note that if $v_i$ has degree 2 (in the tree $T'$), it may not lie in $M$.  See also Figure~\ref{fig:tree}. 
	We claim that:
	\begin{equation}
	    \label{eq:tree-marked} L_i\cap L(D)\sse L_{v_i}\cap L(D)
	\end{equation}
	In other words, a path $P_j$ containing $i$ and a vertex $w$ in $D$ must contain $v_i$ as well. Indeed, the last $T'$-vertex in the path  from $w$ to $i$ must be $v_i$ (the closest vertex to $i$ in $T'$). 
	 We now consider two cases: 
	\begin{itemize}
	    \item If $v_i\in M$, we   set $R_i=\{v_i\}$. By \eqref{eq:tree-marked} we have  $L_i\cap L(D)\sse L_{v_i} =L(R_i) $. 
	    \item 	If $v_i\not\in M$ then $v_i$ must be a degree-2 vertex in $T'$. Let $a_i$ and $b_i$ be the first two  vertices of $M$ that we  encounter if we move  from $v_i$ (along the sub-tree $T'$) in both directions. Set $R_i :=\{a_i, b_i\}$. Let $Q$ be the path from $a_i$ to $b_i$ in $T'$. Observe that $Q$ contains $v_i$,  all internal vertices in $Q$ have degree 2 (in $T'$) and $Q\cap M=\{a_i,b_i\}$. Let $P_j$ be any path which contains $i$ and a vertex $w$ in $D$. By \eqref{eq:tree-marked} $v_i\in P_j$. The part of $P_j$ from $v_i$ to $w$ must lie in $T'$ and hence contains either $a_i$ or $b_i$. 
	\end{itemize}
	Since $|R_i| \leq 2,$ the desired result follows. 
	  \end{proof}
\begin{figure}[htb]
     \centering
     \includegraphics[width=0.95\textwidth]{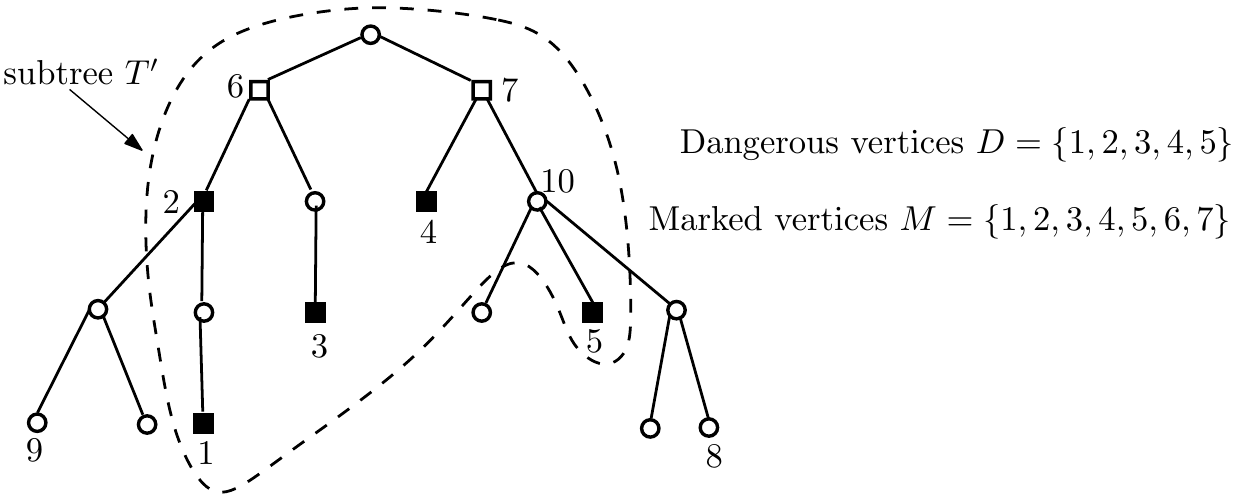}
     \caption{The solid-square vertices are the ``dangerous" vertices $D$. The box vertices are the additional marked vertices  $M\setminus D$. For vertex $8$, we have  $v_8=10$ and  $R_{8}=\{5,7\}$. Similarly, for vertex $9$, $v_9=2$ and $R_{9}=\{2\}$. }
     \label{fig:tree}
 \end{figure}
 
We now consider the $\alpha$-\packable property. As in the case of the line graph application, this is equivalent to bounding the integrality gap of the unsplittable flow problem on trees where vertices have capacities. An analogous result with edge capacities was given by Chekuri et al.~\cite{ChekuriMS07}, and our rounding algorithm is inspired by their approach.  

Consider an instance of the unsplittable flow problem where every vertex in the tree has capacity $\theta$, and path $P_j$ has reward $r_j$ and size $s_j$ (we assume that $\theta \geq \max_j s_j$). Our goal is to find a maximum reward subset of paths which obey the vertex capacities---we call this problem \UFP. It is easy to see  that~\eqref{eq:det_LP} is the natural LP relaxation for this problem.

\begin{lemma}\label{lem:UFtree_LP}
The LP relaxation~\eqref{eq:det_LP} for \UFP has constant integrality gap, and so the above set system is $O(1)$-\packable. 
\end{lemma}
\begin{proof}
Consider a feasible solution $\{y_j\}_{j=1}^n$ to~\eqref{eq:det_LP}.  We root the tree $T$ arbitrarily and this naturally defines an ancestor-descendant relationship on the vertices of the tree. The depth of a vertex is its distance from the root. For each path $P_j$, let $v_j$ be the vertex in $P_j$ with the least depth, and define the {\em depth of $P_j$} to be the depth of $v_j$. 

We partition the set of paths into types: $\cP_s$, the small paths, are the ones with $s_j \leq \theta/2$, and $\cP_l$, the large paths, are the ones with $s_j > \theta/2$. We maintain two feasible sets of paths, $\cS_s\sse \cP_s$ and $\cS_l\sse \cP_l$. We initialize both $\cS_s,\cS_l=\emptyset$. We consider the paths in ascending order of depth. Each path $P_j$ is  rejected immediately with probability $1-y_j/4$ and with the remaining $\frac{y_j}{4}$ probability we do the following: if $P_j$ is a small (resp. large) path, we add it to $\cS_s$ (resp. $\cS_l$) provided the resulting set $\cS_s$ (resp. $\cS_l$) is feasible, i.e., it does not violate any vertex capacity. 
Finally, we  return the better among the two solutions $\cS_s$ and $\cS_l$. 

For the analysis, we will show that 
\begin{equation}\label{eq:tree-path-prob}
\P\left[P_j\in \cS_s\bigcup \cS_l\right] \,\ge \, \frac{y_j}{8},\quad \forall j\in[n]. 
\end{equation}
This would imply the lemma because our solution's expected objective is:
$$\E\left[\max\left\{\sum_{j: P_j\in \cS_s} r_j , \sum_{j: P_j\in \cS_l} r_j  \right\}\right]  \ge \frac12 \sum_{j=1}^n r_j\cdot \P\left[P_j\in \cS_s\bigcup \cS_l\right]\ge \frac{1}{16}\sum_{j=1}^n r_j\cdot y_j. $$
 
 We begin with a key observation, which is easy to see. 
\begin{observation}\label{obs:vj}
Suppose that  path $P_k$ is considered before another path $P_j$  and  $P_j \cap P_k \neq \emptyset$. Then $v_j \in P_k$. 
\end{observation}

\begin{observation}\label{cl:v_j}
   Let $P_j$ be a small(resp. large)  path. 
   Before   path $P_j$ is considered, the load on any vertex $v\in P_j$ due to paths in $\cS_s$ (resp. $\cS_l$) is at most the load due to these paths on $v_j$.
\end{observation}
\begin{proof}
   Assume $P_j$ is a small path (the argument for large paths is identical). 
   Consider a time during the rounding algorithm before $P_j$ is considered. 
   For a vertex $v\in P_j$, let $F_v$ be the set of 
   paths in $\cS_s$ that contain $v$. 
   By  Observation~\ref{obs:vj}, any  path in  $F_v$ also contains $v_j$. This implies the claim. 
  \end{proof}
Observation~\ref{cl:v_j} implies that if we want to check whether adding a path $P_j$ will violate feasibility (of $\cS_s$ or $\cS_l$), it   suffices to check the corresponding load on $v_j$ (as all capacities are uniform). We are now ready to prove \eqref{eq:tree-path-prob}. For any path $P_k$ (small or large), let $I_k$ be  the indicator of the event that $P_k$ does not get immediately rejected; so $\P[I_k] = y_k/4$. We consider two cases: 
\begin{itemize}
    \item  $P_j$ is small. 
 We condition on the event $I_j=1$: note that $\P[P_j\in \cS_s] = \P[I_j=1] \cdot \P[P_j\in \cS_s | I_j=1]$. 
    Let $L' \sse [n]$ denote the indices of paths $P_k$ considered before $P_j$ with $v_j\in P_k$ and $I_k=1$. Note that $L'\sse L_{v_j}$. If the total size of $L'$ is at most $\theta-s_j$, then $P_j$ will get added to $\cS_s$ (conditioned on $I_j=1$).  So, 
\begin{align*}
&\P[P_j \notin \cS_s|I_j=1]  \leq \P[s(L') \geq \theta - s_j] = \P\left[\sum_{k\in L_{v_j}}{s_kI_k}\ge \theta-s_j\right]\\
&\le \frac{\E[\sum_{k\in L_{v_j}}{s_k I_k}]}{\theta-s_j}=\frac{\sum_{k\in L_{v_j}}s_k(y_k/4)}{\theta-s_j} \le \frac{\theta/4}{\theta-\theta/2}= \frac{1}{2},
\end{align*}
where the last inequality follows from LP constraints in~\eqref{eq:det_LP} and the fact that $P_j$ is small. 
Therefore, 
$$ \P[P_j \in \cS_s] = \P[P_j \in \cS_s|I_j=1] \cdot \P[I_j=1] \geq y_k/8. $$

\item $P_j$ is large. Let $L''$ denote the indices of the large paths $P_k$ considered before $P_j$ with $v_j\in P_k$. If none of the paths indexed $L''$ is selected then $P_j$ will be added to $\cS_l$.  Moreover, path $P_k$ can be selected only if $I_k=1$. So, 
\begin{align*}
\P\left[P_j \notin \cS_l|I_j=1\right] & \leq \P\left[\sum_{k\in L''}{I_k}\ge 1\right] \le \sum_{k\in L''}\P[I_k=1] \leq \sum_{ k \in L''} \frac{y_k}{4} \\
& \le \frac{1}{2} \sum_{ k \in L''} \frac{s_k y_k}{\theta}\leq \frac{1}{2},
\end{align*}
where the second last inequality follows from the fact that $s_k \geq \theta/2$ for all $k \in L''$, and the last inequality follows from the fact that $L'' \subseteq L_{v_j}$ and the  LP constraints~\eqref{eq:det_LP}. As in the previous case, this implies $ \P[P_j \in \cS_l] \ge \frac{y_j}8$. 
\end{itemize}
This completes the proof of \eqref{eq:tree-path-prob} and the lemma. 
  \end{proof} 

Combining Theorem~\ref{thm:main} with Lemmas~\ref{lem:UFtree_LP} and \ref{lem:UFtree_set}, we get 
\begin{corollary}\label{thm:UFtree}
 	There is an $O(\log\log m)$-approximation algorithm for \makespan
 	when the resources are given by the vertices in a tree and the tasks are given by paths in this tree.
  \end{corollary}

\subsection{Axis-Aligned Rectangles in the Plane}\label{subsec:rect}
We now consider the following geometric set system: the tasks are  $n$ axis-aligned rectangles in the  plane and the resources are all points in the plane. The set $L_i$ for a resource (i.e., point) $i$ is given by the set of rectangles containing $i$. Note that any set of $n$ rectangles partitions the plane into $poly(n)$ many connected regions: this follows from the fact that the total number of intersection points is $O(n^2)$. We designate one point in each connected region as the representative point for that region. Clearly, it suffices to bound the loads on the representative points. Note that the number of representative points is $m=poly(n)$. Below, whenever we refer to an arbitrary point $p$, it  is equivalent to using $p$'s representative point.

\begin{lemma}\label{lem:RecPack_set}
The above mentioned set-system is 4-safe. 
\end{lemma}
\begin{proof}
Let $D = \{(x_i,y_i)\}_{i=1}^k$ be a  subset of  points. Define the set $M := \ext(D)$ to be the Cartesian product of all the $x$ and $y$ coordinates in $D$, i.e., $M = \{(x_i, y_j): (x_i, y_i), (x_j, y_j) \in D\}$. Clearly, $|M| \leq k^2$, which satisfies the first condition in the definition of $\lambda$-\safe. Notice that the points in $M$ correspond to a rectangular grid ${\cal G}$  partitioning the plane, where  the rectangles on the boundary of ${\cal G}$ are unbounded. See Figure~\ref{fig:rect-fat}(a).

Let $p$ be any point. We need to define a set $R_p \subseteq M$ such that $L_p \cap L(D) \subseteq L(R_p)$.  
 Let $Q$ denote the  minimal  rectangle in the grid ${\cal G}$ that contains $p$. Let $R_p\sse M$ denote the corners of rectangle $Q$ (if $Q$ is unbounded then it has fewer than four corners, but the following argument still applies.)  Define $R_p$ to be the set of these corner points. Now let $J$ be a task (i.e., rectangle) containing $p$ and a point in $D$. By construction of $M$, it must be that $J$ contains  one of the points in $R_p$. This proves the lemma. 
  \end{proof}

We now consider the $\alpha$-\packable assumption.  Corollary~\ref{cor:rectangle-pack} in Appendix~\ref{app:packing}  proves that this set-system is $O((\log\log n)^2)$-\packable. Therefore, using Theorem~\ref{thm:main} we obtain:
\begin{corollary}
     There is an $O\left((\log\log n)^3\right)$-approximation algorithm for \makespan when
     the resources are represented by all     points in the plane
     and the tasks are given by a set of $n$ axis-aligned rectangles.
\end{corollary}

 \begin{figure}[t]
     \centering
     \includegraphics[width=0.99\textwidth]{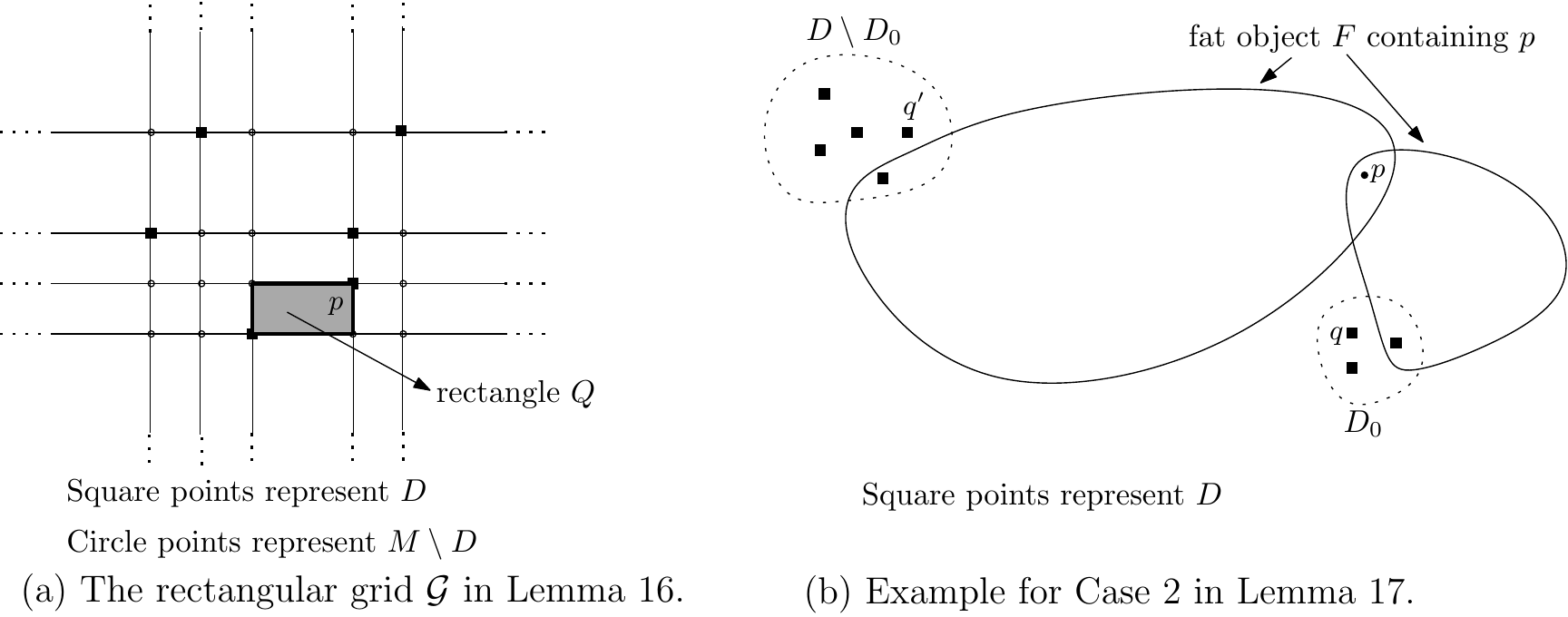}
     \caption{Examples for rectangles and fat objects. \label{fig:rect-fat}}
     \label{fig:fat}
 \end{figure}

\subsection{Fat Objects in the Plane}
We now consider more general shapes which are not skewed in any particular direction. The tasks are 
given by a set of $n$  ``fat'' objects in a plane and the resources are given by the set of  all points in the plane. We assume that the number of intersection points between any pair of objects is constant, which  is true for all our specific applications (disks, triangles, rectangles). This implies that any set of $n$ objects 
partitions the plane into $m=poly(n)$ many connected regions. As in \S\ref{subsec:rect}, we designate one point in each connected region as the representative point for that region and focus on  the loads of the $m$ representative points. Whenever we refer to an arbitrary point $p$, it  is equivalent to using $p$'s representative point.   
For any resource (i.e., point) $p$, $L_p$ is the set of fat objects containing $p$. 
\begin{definition}[Fat objects \cite{ChanH12}]\label{defn:fat}
A set  
${\cal F}$ of objects in $\R^2$ is called fat if for every axis-aligned square $B$ of side-length $r$, we can find  a constant number of points $Q(B)$ such that every object in ${\cal F}$ that intersects $B$ and has diameter at least $r$ also contains some point in $Q(B)$. 
\end{definition}
Examples of fat objects include squares/disks (with arbitrary diameters) and  triangles/rectangles with constant aspect ratio (i.e., when the ratio of the maximum to minimum side length is constant). For concreteness, one can consider all objects to be disks; note that the radii can be different. 
   
\def\D{{\cal H}}

\begin{lemma}\label{lem:RecFat_set}
The above-mentioned set system is $O(1)$-\safe. 
\end{lemma}
\begin{proof}
Let $\cal F$ denote the set of fat objects represented by the tasks. 
Let $D$ be any subset of points in the plane and $\D$ denote the set of all non-zero pairwise distances between the points  in $D$; note that $|\D|\le |D|^2$. 

We define the set $M:=\ext(D)$ as follows: 
for each point $p\in D$ and distance $\theta\in \D$ let $G(p,\theta)$ be the square centered at $p$ with side-length $10\theta$. We divide this square into a grid consisting of smaller squares (called {\em cells}) of side length $0.1 \theta$. So $G(p, \theta)$ has 100 cells in it.
  For each cell $B$ in  $G(p,\theta)$, add  to $M$  the points $Q(B)$ from  Definition~\ref{defn:fat} with $r:=0.1 \theta$. 

Clearly, $|M|\le O(1)\cdot |D|\, |\D| = O(|D|^3)= poly(|D|)$ as required by the first condition of $\lambda$-\safe. We now check the second condition of this definition. Let $p$ be an arbitrary point. We need to show that there is a constant size subset $R_p \subseteq M$ such that $L_p \cap L(D) \subseteq L(R_p) $. 
Let $q$ be the closest point in $D$ to $p$, and $d(p,q)$ denote the (Euclidean) distance between these two points. Note that $d(p,q)$ may not belong to $\D$. 
We consider the following cases:

\noindent {\em Case 1: there exists some $\theta \in \D$ with $\frac{d(p,q)}{5} \leq \theta \leq 5d(p,q)$.} Consider the grid $G(q, \theta)$. There must be some cell $B$ in this grid that contains $p$. Define $R_p := Q(B)$, where $Q(B)$ is as in Definition~\ref{defn:fat} (with respect to $\cal F$). 

Let us see why this definition has the desired property. Let $F \in {\cal F}$ be any object which contains $p$ and some point $r\in D$. Since $q$ is the closest point in $D$ to $p$, the diameter of $F$ is at least $d(p,r)\ge d(p,q) > 0.1 \theta$, which is the side length of $B$. Note also that $F$ intersects $B$ because $p \in F$. So, by Definition~\ref{defn:fat}, the object $F$ must intersect $Q(B)$ as well. Thus, $L_p \cap L(D) \subseteq L(R_p)$.

\noindent{\em Case 2: there is no $\theta \in \D$ with $\frac{d(p,q)}{5} \leq \theta \leq 5d(p,q)$.}  Let $D_0\sse D$ be the subset of  $D$ at distance at most $d(p,q)/5$ from $q.$ Let $q'$ be the point in $D \setminus D_0$ which is closest to $p$; see Figure~\ref{fig:rect-fat}(b). (If $D\setminus D_0=\emptyset$ then we just ignore all steps involving $q'$ below.) Since $q' \notin D_0$, $d(q,q') > d(p,q)/5$. Moreover, as  $\D\cap [\frac{d(p,q)}{5}, 5d(p,q)]=\emptyset$ we have $d(q,q')>5d(p,q)$. Using triangle inequality, we get $d(p,q) + d(p,q') \geq d(q,q') > 5 d(p,q), $ and so, $d(p,q') > 4 d(p,q).$ We are now ready to define $R_p$. There are two kinds of points in $R_p$:
\begin{itemize}
    \item Type-1 points: If $D_0$ is the singleton set $\{q\}$, add $q$ to $R_p$. Otherwise, let $\Delta \in \D$ be maximum pairwise distance between any two points in $D_0$. Note that:
        $$\Delta = \max_{q_1,q_2\in D_0} d(q_1,q_2)\le  \max_{q_1,q_2\in D_0} \left( d(q,q_1) + d(q,q_2)\right)\le \frac{2}{5}d(p,q).$$
        For each cell $B$ in the grid $G(q, \Delta)$, add $Q(B)$ to $R_p$. Note that the number of cells is 100, and so we only add $O(1)$ many points to $R_p$. 
    \item Type-2 points: Recall that $d(p,q') > 4 d(p,q)$. It follows that $d(q,q') \leq d(p,q) + d(p,q') \leq 1.25 d(p,q')$, and $d(q,q') \geq d(p,q') - d(p,q) \geq 0.75 d(p,q')$. So there is an element $\theta' \in \D$ with $0.75 d(p,q')\le \theta' \le 1.25 d(p,q')$. We consider the grid $G(q',\theta')$ -- there must be a cell in this grid which contains $p$. Let $B$ be this cell. Add all the points in $Q(B)$ to $R_p$. Again, we  only add a constant number of points to $R_p$. 
\end{itemize}

It is clear that $R_p$ is a subset of $M$. Now, consider any object $F \in {\cal F}$ which contains $p$ and some point in $D$. We will show that $F$ also contains some point in $R_p$, which would prove $L_p \cap L(D) \subseteq L(R_p)$. 
Two cases arise:
    \begin{itemize}
        \item $F \cap D_0 \neq \emptyset$: If $D_0=\{q\}$, then $F$ clearly intersects $R_p$. So assume that $|D_0|\ge 2$. Recall that $\Delta$ is the diameter of $D_0$. 
        So,  
        the grid $G(q, \Delta)$ contains all of $D_0$, which implies that there is a cell $B$ in $G(q, \Delta)$ intersecting $F$. 
As $q$ is the closest point in $D$ to $p$, the diameter of $F$ is at least $d(p,q) \geq 0.1 \Delta$,  the side length of $B$. Hence, by Definition~\ref{defn:fat}, $F$ must contain a point in $Q(B)$, and so, contains one of the type-1 points in $R_p$. 
        \item $F \cap D_0 = \emptyset$: Recall point $q'$ and value $\theta'$ used in the definition of type-2 points in $R_p$.  Note that there is some cell $B$  in $G(q',\theta')$ that contains $p$; so object $F$ intersects cell $B$. Further, $F$ contains some point $r\in D\setminus D_0$ which implies that the diameter of $F$ is at least $d(p,r) \ge d(p,q')\ge 0.8\cdot \theta'$, which is larger than the side length of $B$. So, by Definition~\ref{defn:fat},  $F$ must contain a point in $Q(B)$, i.e., some type-2 point in $R_p$. 
    \end{itemize}
This completes the proof of the lemma. 
  \end{proof}

For the  $\alpha$-\packable condition, Corollary~\ref{cor:disk-pack} in Appendix~\ref{app:packing} implies that disks (of arbitrary radii) are $O(\log\log n)$-\packable. And, Corollary~\ref{cor:triangle-pack} implies that fat triangles are $O(\log^*n \cdot \log\log n)$-\packable. Combined with Theorem~\ref{thm:main} and Lemma~\ref{lem:RecFat_set}, we obtain:

\begin{corollary}\label{thm:fat}
The \makespan problem admits
an $O\left((\log\log n)^2)\right)$-approximation algorithm when tasks  are   disks in the plane, and 
an $O\left((\log^*n)\cdot (\log\log n)^2\right)$-approximation algorithm
when tasks  are fat triangles in the plane.
  \end{corollary}

%%%%%%%%%%%%%%%%%%%%%%%%%%%%%%%%%%%%%%%%%%%%%%%%%%%%%%%%%%%%%

\section{Integrality Gap Lower Bounds}
\label{sec:gaps-lowerbounds}
We now study the limitations of our LP relaxation~\eqref{eq:UFLP:assign}--\eqref{eq:UFLP:nonneg}. There are two natural questions  -- (i) can we obtain an $O(1)$-approximation for \makespan under the $\alpha$-\packable and $\lambda$-\safe assumptions with $\alpha,\lambda=O(1)$? and (ii) can we obtain an approximation ratio similar to Theorem~\ref{thm:main} without the $\lambda$-\safe assumption?  
For the first question, we show that even for set systems given by intervals on a line (as in \S\ref{sec:intline}) where $\alpha,\lambda=O(1)$, the integrality gap of our LP is  
 $\Omega(\log^*m)$. 
  For the second question, we show that the integrality gap of our LP for general set systems is $\Omega \left( \frac{\log m}{(\log \log m)^2} \right)$. So we cannot get an approximation ratio that is significantly better than logarithmic without some additional  condition (such as $\lambda$-\safe) on the set system. 
\subsection{Lower Bound for Intervals on a Line}
\label{sec:lbd-line}

We consider the set system as in \S\ref{sec:intline}. Recall that resources are given by $m$ vertices on a line, and tasks by a set of $n$ intervals on the line. We construct such an instance of \makespan with $\Omega(\log^*m)$-integrality gap. 

Let $H$ be an integer. The line consists of $m=2^H$ points and $n=2^{H+1}-1$ intervals. The intervals   are arranged in a binary tree structure. For each ``depth'' $d=0,1,\cdots H$, there are $2^d$ many disjoint depth-$d$ intervals of width $m/2^d$ each. We can view these intervals as nodes in a complete binary tree ${\cal T}$ of depth $H$ where the nodes at depth $d$ correspond to the depth-$d$  intervals, and for any interval $I$ and its parent $I'$ we have $I\sse I'$. Moreover, points in the line correspond to root-leaf paths in ${\cal T}$  where all intervals in the root-leaf path contain the corresponding point. 
The size of every depth-$d$ interval $j$ is a random variable $X_j=\Ber(2^{-d})$, i.e. $X_j=1$ w.p. $2^{-d}$ and  $X_j=0$ otherwise. The target number of intervals is $t=n$: so we need to select {\em all} the intervals. 

Consider the LP relaxation with a target bound of $1$ on the expected makespan. We will show that the LP~\eqref{eq:UFLP:assign}-\eqref{eq:UFLP:nonneg} is feasible with  decision variables $y_j=1$ for all intervals $j$. Note that every random variable $X_j$ is already truncated (there is no instantiation larger than one). So constraints \eqref{eq:UFLP:assign},  \eqref{eq:UFLP:exceptn}  and \eqref{eq:UFLP:nonneg} are clearly satisfied. 
\begin{lemma}\label{lem:line-LB-lp}
For any $K\sse [m]$ with $k=|K|$ we have $\sum_{j\in L(K)} \beta_k(X_j)\le 4k$.  Hence, constraint~\eqref{eq:UFLP:subset} is satisfied with $b=4$ on the right-hand-side.
\end{lemma}
\begin{proof}
Consider any subset $K\sse [m]$ of vertices on the line. Recall that for any vertex $i$, $L_i$ denotes the set of intervals that contain it; and  $L(K):=\bigcup_{i\in K} L_i$. We partition $L(K)$ into the following two sets: $L'$ consisting of intervals of depth at most $\log k$ and $L''=L(K)\setminus L'$ consisting of intervals of depth more than $\log k$. We will bound the summation separately for these two sets.

{\em Bounding the contribution of $L'$.} Note that the total number of intervals of depth at most $\log k$ is less than $2k$. So $|L'|<2k$. Moreover, $\beta_k(X_j)\le 1$ for all intervals $j$. So   $\sum_{j\in L'}\beta_k(X_j)\le |L'| < 2k$.

{\em Bounding the contribution of  $L''$.} Consider any vertex $i\in K$.
For each depth $d=0,\cdots H$, $L_i$ contains exactly one interval of depth $d$.  So we have
\begin{align*}
    \sum_{j\in L''\cap L_i} \beta_k(X_j) &\le \sum_{d=\log k}^H \beta_k(\Ber(2^{-d})) = \frac{1}{\log k}\sum_{d=\log k}^H \log\left(1+(k-1)2^{-d}\right)\\
    & \le  \frac{2(k-1)}{\log k}\sum_{d=\log k}^H 2^{-d}\,\,\le\,\, 2.
\end{align*}
The first inequality used the facts that (i) $L''$ contains only intervals of depth more than $\log k$ and (ii) the size of each  depth-$d$ interval is $\Ber(2^{-d})$. The second inequality uses $\log(1+x)\le 2x$ for all $x\ge 0$. It now follows that $\sum_{j\in L''}\beta_k(X_j)\le \sum_{i\in K}\sum_{j\in L''\cap L_i} \beta_k(X_j) \le 2k$. 

Combining the two bounds above, we obtain the lemma. 
  \end{proof}

Next, we show that the expected makespan when all the $n$ intervals are selected is $\Omega(\log^* n)$. To this end, we will show that with constant probability, there is some root-leaf path in ${\cal T}$ for which $\Omega(\log^* n)$ random variables in it have size one. Define a sequence $\{h_i\}_{i=0}^c$ as follows:
$$h_0=2,\quad h_{i+1}-h_i=h_i\cdot 2^{h_i} \mbox{ for }i=1,\cdots c-1.$$
We choose $c= \Theta(\log^* H)$ so that $h_c\le H$. 
\begin{lemma}\label{lem:line-LB-subtree}
    For any depth-$d$ interval $j$, let ${\cal I}$ denote the intervals in the subtree of ${\cal T}$ below $j$, from depth $d$ to depth $d+d2^d$. Then, 
    $$\P\left[ \sum_{v\in {\cal I}} X_v \ge 1\right] \ge 1-e^{-d}.$$  
\end{lemma}
\begin{proof}
We show that $\P\left[ \sum_{v\in {\cal I}} X_v = 0 \right] \le e^{-d}$, which will imply the desired result.  Note that for each $h=0,\ldots, d2^d$, ${\cal I}$ contains $2^h$ intervals at depth $d+h$ and each of these intervals has  size given by $\Ber(2^{-d-h})$.  By independence, the probability  that all these sizes are zero is:
$$\prod_{v\in {\cal I}} \P[X_v=0] = \prod_{h=0}^{ d2^d} (1-2^{-d-h})^{2^h} \le \prod_{h=0}^{ d2^d} e^{-2^{-d}} = e^{-d},$$
which proves the lemma.   \end{proof}

\begin{lemma}\label{lem:line-LB-int}
    With probability at least $\frac12$, there is a root-leaf path in ${\cal T}$  such that at least $c$ random variables in it are 1. 
\end{lemma}
\begin{proof}
We show the following by induction on $i$, 
$0 \leq i \leq c$: 
\begin{align}
&\mbox{with probability at least $\prod_{i'=0}^{i-1} (1-e^{-h_{i'}})$, there is a depth-$h_i$ node $v_i$ } \label{eq:LB-ind}\\
&\mbox{where the root to $v_i$ path has at least $i$ random variables of value $1$.} \notag 
\end{align}
For $i=0$, this follows easily because the root itself is 1 with probability 1. We now assume the induction hypothesis~\eqref{eq:LB-ind} for some $i < c$ and prove it for $i+1$. Let $V_i$ be the set of nodes (i.e., intervals) in $\cal T$ at depth $h_i$. For an interval $j \in V_i,$ let $E_j$ be the event that $j$ is the first vertex in $V_i$ (say from the left to right ordering) such that the root to $j$ path has at least $i$ random variables which are 1. Let $\cI_j$ be the sub-tree of depth $h_i \cdot 2^{h_i}$ below $j$ (so the leaves of $\cI_j$ are at depth $h_{i+1}$); and $E_j'$ be the event that there is a random variable in $\cI_j$ which is 1.
Lemma~\ref{lem:line-LB-subtree} implies that for any $j \in V_i$, $\P[E_j'] \geq (1-e^{-h_i}).$ Since the events $E_j$ are disjoint, and are independent of $E_{j'}'$ for any $j' \in V_i$, we get 
$$ \P[ \exists j \in V_i: E_j \wedge E_j'] = \sum_{j \in V_i} \P[E_j \wedge E_j'] = \sum_{j \in V_i} \P[E_j] \cdot \P[E_j'] \geq (1-e^{-h_i}) \sum_{j \in V_i} \P[E_j].$$
Since the events $E_j$ are disjoint, $\sum_j \P[E_j] = \P[ \exists j \in V_i: E_j]$. By \eqref{eq:LB-ind}, this probability is at least $\prod_{i'=0}^{i-1} (1-e^{-h_{i'}})$. Hence, 
$$ \P[ \exists j \in V_i: E_j \wedge E_j'] \ge \prod_{i'=0}^{i} (1-e^{-h_{i'}}).$$
Note that if events $E_j$ and $E'_j$ are true (for any $j\in V_i$) then there is some depth-$h_{i+1}$ node $v_{i+1}$ with at least $i+1$ r.v.s of value $1$ on the root to $v_{i+1}$ path. This proves the inductive statement for $i+1$. Finally, using \eqref{eq:LB-ind} for $i=c$, the probability that there is a root-leaf path with at least $c$ r.v.s of value one is at least 
 $\prod_{i=0}^{c-1}(1-e^{-h_i})\ge 1-\sum_{i=0}^{c-1}e^{-h_i}\ge \frac12$.
  \end{proof}

Combining Lemmas \ref{lem:line-LB-lp} and \ref{lem:line-LB-int}, we obtain:

\begin{theorem}
The integrality gap of the LP~\eqref{eq:UFLP:assign}--\eqref{eq:UFLP:nonneg}  for \makespan  when the set system is given by intervals on the line is  $\Omega(\log^*m)$.
\end{theorem}

\subsection{Lower Bound for General Set Systems}
\label{sec:lbd-general}
Now we consider \makespan for general set systems and show that the LP relaxation has $\Omega(\frac{\log m}{(\log\log m)^2})$ integrality gap.

The instance consists of $n=q^2$  tasks and $m=q^q$ resources where $q$ is some parameter. For each task $j$, the random variable $X_j$ is a Bernoulli random variable that takes value 1 one with probability $\frac1q$, i.e., the distribution of $X_j$ is $\Ber(\frac{1}{q})$. The tasks are partitioned into $q$ groups -- $T_1, \cdots T_q$, with $q$ tasks in each group. Each resource is associated with a choice of one task $a_j$ from each group $T_j$, $j \in [q]$. In other words, the set $L_i$ for any resource $i$ has cardinality $q$ and contains exactly one element from each of the groups $T_j$. Thus, the total number of resources is $q^q$.  
The target number of tasks to be chosen is $n=q^2$, which means every task must be selected.

We first observe that the expected makespan is $\Omega(q)$. Indeed, consider any group $T_j$. With probability $1-(1-1/q)^q\approx 1-1/e$, there is a task $a_j \in T_j$ for which the random variable $X_{a_j}$ is 1. So the expected number of groups for which this event happens is about $(1-1/e)q$. As there is a resource associated with {\em every} choice of one task from each group, the expected makespan  is at least $(1-1/e)q$.

Consider the LP relaxation with a target bound of $B=\log q$ on the expected makespan. We will show that the LP constraints~\eqref{eq:UFLP:assign}-\eqref{eq:UFLP:nonneg} are feasible with  decision variables $y_j=1$ for all objects $j$. We will scale all the random variables down by a factor of $B$ (because the LP relaxation assumes that the target makespan is 1). 
Let $X$ denote the scaled Bernoulli r.v. with $X=\frac1B$ w.p. $\frac1q$ and $X=0$ otherwise. Since this random variable will never exceed 1, $X'$ (the truncated part) is same as $X$, and $X''$ (the exceptional part) is 0.  So constraints \eqref{eq:UFLP:assign},  \eqref{eq:UFLP:exceptn}  and \eqref{eq:UFLP:nonneg} are clearly satisfied.  Moreover,
\begin{equation}
    \label{eq:gen-LB-beta} 
    \beta_k(X) \le \frac{1}{\log k}\log\left(1+\frac{k^{1/\log q}}{q}\right) \le   \frac{2k^{1/\log q}}{q\, \log k}, 
\end{equation}
where we used $\log(1+x)\le 2x$ for all $x\ge 0$.

\begin{lemma}
    Constraint~\eqref{eq:UFLP:subset} is satisfied with $b=2e^2$ for the above instance. 
\end{lemma}
\begin{proof}
Consider any subset $K\sse [m]$ of $k=|K|$ resources. Recall that $L(K)\sse [n]$ denotes the subset of tasks contained in any of the sets corresponding to $K$. As every random variable has the same  distribution as $X$, the left-hand-side (LHS) in \eqref{eq:UFLP:subset} is just $|L(K)|\cdot \beta_k(X)$. We now consider three cases:
\begin{itemize}
    \item $k\le q$. We have $|L(K)|\le kq$ as each resource is loaded by exactly $q$ tasks.  Using \eqref{eq:gen-LB-beta}, the LHS is at most $kq\cdot \beta_k(X)\le kq \frac{2q^{1/\log q}}{q\, \log k} \le 2e \cdot k$.
    \item $q<k\le q^2$. We now use  $|L(K)|\le n = q^2$. By  \eqref{eq:gen-LB-beta}, we have $\beta_k(X)\le \frac{2\cdot q^{2/\log q}}{q\, \log k} \le 
    \frac{2e^2}{q\, \log k}$. So $LHS\le q^2\cdot \beta_k(X)\le \frac{2e^2 q}{\log k}\le 2e^2\cdot k$.
    \item $k>q^2$. Here we just use $|L(K)|\le q^2$ and $\beta_k(X)\le 1$ to get $LHS\le k$.
\end{itemize}
The lemma is proved as $LHS\le 2e^2\cdot k$ in all cases.
  \end{proof}

As $q=\Theta(\frac{\log m}{\log \log m})$, the LP integrality gap  is $\Omega(\frac{q}{\log q}) = \Omega\left( \frac{\log m}{(\log \log m)^2}\right)$.

We also observe that the integrality gap of LP \eqref{eq:det_LP}  is $\alpha\le 2$ for all instances of the deterministic problem that need to be solved in our algorithm.   As all  sizes are identically distributed, we only need to consider deterministic instances of the reward-maximization problem for which  all (deterministic) sizes are identical, say $s$.  
  Using the structure of the above set-system, it is clear that an optimal LP solution will assign the same value $z_i\in [0,1]$ to all tasks in any group $G_i$. So the LP objective equals $\sum_{i=1}^q r(G_i)\cdot z_i$ where $r(G_i)$ is the total reward of the tasks in $G_i$. The constraints in \eqref{eq:det_LP} imply $\sum_{i=1}^q z_i\le \frac{\theta}{s}$. So this LP now reduces to the max-knapsack problem, which is known to have integrality gap at most two. In particular, choosing all tasks in the $\lfloor \frac{\theta}{s} \rfloor$ groups $G_i$ with the highest $r(G_i)$ yields total reward at least half the LP value.
 
 \section{Conclusion}
We considered a class of stochastic makespan minimization problems, where a specific number of tasks need to be selected and each selected task induces a random load on multiple resources. When the set-system (consisting of the tasks and resources) satisfies some geometric properties, we obtained good approximation algorithms. In particular, for stochastic intervals on a line, we obtained an $O(\log \log m)$-approximation algorithm. Our approach was based on a natural LP relaxation, which also has integrality gap $\Omega(\log^*m)$. 
Finding the correct integrality gap of this LP remains an interesting open question. Obtaining a constant-factor approximation (or hardness results) for stochastic intervals is another  interesting  direction.


\bibliographystyle{spmpsci}      % mathematics and physical sciences
\bibliography{smm}   % name your BibTeX data base

\begin{thebibliography}{10}
\providecommand{\url}[1]{{#1}}
\providecommand{\urlprefix}{URL }
\expandafter\ifx\csname urlstyle\endcsname\relax
  \providecommand{\doi}[1]{DOI~\discretionary{}{}{}#1}\else
  \providecommand{\doi}{DOI~\discretionary{}{}{}\begingroup
  \urlstyle{rm}\Url}\fi

\bibitem{AgarwalM06}
Agarwal, P.K., Mustafa, N.H.: Independent set of intersection graphs of convex
  objects in 2d.
\newblock Comput. Geom. \textbf{34}(2), 83--95 (2006)

\bibitem{APS08}
Agarwal, P.K., Pach, J., Sharir, M.: State of the union of geometric objects.
\newblock In: {Surveys in Discrete and Computational Geometry Twenty Years
  Later}, pp. 9--48 (2008)

\bibitem{AronovBES14}
Aronov, B., de~Berg, M., Ezra, E., Sharir, M.: Improved bounds for the union of
  locally fat objects in the plane.
\newblock {SIAM} J. Comput. \textbf{43}(2), 543--572 (2014)

\bibitem{CarrV02}
Carr, R.D., Vempala, S.S.: Randomized metarounding.
\newblock Random Struct. Algorithms \textbf{20}(3), 343--352 (2002)

\bibitem{chakrabarti2007approximation}
Chakrabarti, A., Chekuri, C., Gupta, A., Kumar, A.: Approximation algorithms
  for the unsplittable flow problem.
\newblock Algorithmica \textbf{47}(1), 53--78 (2007)

\bibitem{ChalermsookC09}
Chalermsook, P., Chuzhoy, J.: Maximum independent set of rectangles.
\newblock In: {SODA}, pp. 892--901 (2009)

\bibitem{ChalermsookW21}
Chalermsook, P., Walczak, B.: Coloring and maximum weight independent set of
  rectangles.
\newblock In: Proceedings of the {ACM-SIAM} Symposium on Discrete Algorithms,
  pp. 860--868 (2021)

\bibitem{Chan04}
Chan, T.M.: A note on maximum independent sets in rectangle intersection
  graphs.
\newblock Inf. Process. Lett. \textbf{89}(1), 19--23 (2004)

\bibitem{ChanH12}
Chan, T.M., Har{-}Peled, S.: Approximation algorithms for maximum independent
  set of pseudo-disks.
\newblock Discrete {\&} Computational Geometry \textbf{48}(2), 373--392 (2012)

\bibitem{ChekuriMS07}
Chekuri, C., Mydlarz, M., Shepherd, F.B.: Multicommodity demand flow in a tree
  and packing integer programs.
\newblock {ACM} Trans. Algorithms \textbf{3}(3), 27 (2007)

\bibitem{ChekuriVZ10}
Chekuri, C., Vondr{\'{a}}k, J., Zenklusen, R.: Dependent randomized rounding
  via exchange properties of combinatorial structures.
\newblock In: 51th Annual {IEEE} Symposium on Foundations of Computer Science,
  {FOCS} 2010, October 23-26, 2010, Las Vegas, Nevada, {USA}, pp. 575--584
  (2010)

\bibitem{cornuejols1977exceptional}
Cornuejols, G., Fisher, M.L., Nemhauser, G.L.: Location of bank accounts to
  optimize float: An analytic study of exact and approximate algorithms.
\newblock Management science \textbf{23}(8), 789--810 (1977)

\bibitem{ElwaM}
Elwalid, A.I., Mitra, D.: Effective bandwidth of general markovian traffic
  sources and admission control of high speed networks.
\newblock IEEE/ACM Transactions on Networking \textbf{1}(3), 329--343 (1993)

\bibitem{gupta2018stochastic}
Gupta, A., Kumar, A., Nagarajan, V., Shen, X.: Stochastic load balancing on
  unrelated machines.
\newblock Math. Oper. Res. \textbf{46}(1), 115--133 (2021)

\bibitem{Hui}
Hui, J.Y.: Resource allocation for broadband networks.
\newblock IEEE J. Selected Areas in Comm. \textbf{6}(3), 1598--1608 (1988)

\bibitem{Kelly-notes}
Kelly, F.P.: Notes on effective bandwidths.
\newblock In: Stochastic Networks: Theory and Applications, pp. 141--168.
  Oxford University Press (1996)

\bibitem{KRT}
Kleinberg, J., Rabani, Y., Tardos, E.: Allocating bandwidth for bursty
  connections.
\newblock SIAM J. Comput. \textbf{30}(1), 191--217 (2000)

\bibitem{Molinaro19}
Molinaro, M.: Stochastic $\ell_p$ load balancing and moment problems via the
  l-function method.
\newblock In: {SODA}, pp. 343--354 (2019)

\bibitem{Srinivasan99}
Srinivasan, A.: Improved approximation guarantees for packing and covering
  integer programs.
\newblock {SIAM} J. Comput. \textbf{29}(2), 648--670 (1999)

\end{thebibliography}

\appendix

\section{Scaling  the Optimal Value}
\label{app:scaling}
Suppose that ${\cal A}$ is a polynomial algorithm that given any \smm instance,  returns  one of the following:
\begin{itemize}
    \item a solution of objective at most $\rho$, or
    \item a certificate that  the optimal  value is more than $1$.
\end{itemize}
Using this, we  provide a polynomial time $O(\rho)$-approximation algorithm for \smm. Observe that the optimal value $\opt$  of \smm lies between $L:=\min_{j\in [n]} \E[X_j]$ and $U:=n\cdot \max_{j\in [n]} \E[X_j]$. It follows that $\frac{B^*}2\le \opt\le B^*$ for some value $B^*$ in the set:
$${\cal G} := \left\{ 2^\ell \cdot L:  0 \le \ell \le \log_2 (U/L) +1, \, \ell\in \mathbb{Z}\right\}$$
 For each $B\in {\cal G}$, consider the  modified \smm instance with r.v. $X_j/B$ for each task $j$; and  run  algorithm ${\cal A}$ on this instance. Finally, return the solution with the smallest objective obtained over all $B\in {\cal G}$. Note that when $B=B^*$, the optimal value of the modified \smm instance is at most $1$: so  algorithm ${\cal A}$ must find a solution with (modified) objective at most $\rho$, i.e., the expected makespan under the original r.v.s $\{X_j\}$ is at most $\rho\cdot B^*\le 2\rho \cdot \opt$. The number of times we call algorithm ${\cal A}$ is $O(\log (U/L))$. Note that  $L\ge s_{min}$ and  $U\le n\cdot s_{max}$ where $s_{min}$ and $s_{max}$ are the minimum and maximum values that the r.v.s take. So, $\log (U/L)  = \log(n\frac{s_{max}}{s_{min}})$, which is polynomial in the instance size. Hence, we obtain a polynomial time $2\rho$-approximation algorithm for \smm. 

\section{The $\alpha$-Packable Property for Rectangles and Fat Objects}\label{app:packing}
In this section, we relate the $\alpha$-packable property of a set system to the the intergrality gap of the natural LP relaxation for maximum (weighted) independent set for the set system. Using known integrality gap results for maximum independent set for axis-parallel rectangles and fat objects, we can show $\alpha$-packability of the corresponding set systems for suitable values of $\alpha$.

Recall the setting in the {\em $\alpha$-\packable} property. There is a set
      system $([n], \cL)$ with  size $s_j \geq 0$ and reward $r_j$ for each
      element $j \in [n]$, and a  bound $\theta \geq
      \max_j s_j$. We are interested in the integrality gap (and a polynomial-time rounding algorithm) for   LP~\eqref{eq:det_LP}, restated below.    
    \begin{equation*} 
    \max \bigg\{ \sum_{j \in [n]} r_j\cdot y_j \,:\, \sum_{j\in L}
    s_j\cdot y_j \le \theta, \, \, \forall L \in \cL; \   0\le y_j\le 1, \, \, \forall j \in [n] \bigg\}. 
    \end{equation*}
When all  sizes $s_j=1$ and the bound $\theta=1$, we obtain the {\em independent set} LP:
    \begin{equation} \label{eq:ind-set-LP}
    \max \bigg\{ \sum_{j \in [n]} r_j\cdot y_j \,:\, \sum_{j\in L}
    y_j \le 1, \, \, \forall L \in \cL; \   0\le y_j\le 1, \, \, \forall j \in [n] \bigg\}. 
    \end{equation}
Note that  the corresponding integral problem involves selecting a max-reward subset of {\em disjoint} elements. (Elements $e$ and $f$ are disjoint if there is no set $L\in \cL$ with $e,f\in L$.) 
\begin{theorem}\label{thm:pack-indset}
Suppose that the  independent set LP~\eqref{eq:ind-set-LP} has integrality gap $\rho$ and an associated polynomial time 
rounding algorithm. Then, the set-system is $O(\rho\cdot \log\log m)$-\packable. 
\end{theorem}
\begin{proof}
The proof proceeds in several steps: (i) we first consider the special case when all $s_j$ values are 1, but the parameter $\theta$ can be arbitrary,  (ii) secondly, when $\theta \gg s_j$ for all $j$ (by more than a $\log m$ factor), we use randomized rounding, and (iii) finally, for the general case, we use a standard bucketing trick to create $O(\log \log m)$ groups, and show that one of the above two steps will work for each of these groups.  

We give details of the first step. 
We show  a rounding algorithm  for the following LP, and show that its integrality gap is at most $2 \rho$:
    \begin{equation} \label{eq:b-ind-set-LP}
    \max \bigg\{ \sum_{j \in [n]} r_j\cdot y_j \,:\, \sum_{j\in L}
    y_j \le b, \, \, \forall L \in \cL; \   0\le y_j\le 1, \, \, \forall j \in [n] \bigg\}. 
    \end{equation}
    Here, we assume that $b\ge 1$ is integer. 
Note that this is a special case of the LP~\eqref{eq:det_LP} used in the $\alpha$-\packable condition. 

\vspace*{2 mm}
\noindent
{\bf{(i) Rounding for the LP \eqref{eq:b-ind-set-LP}: }}
\quad We combine the rounding algorithm for the independent set LP relaxation~\eqref{eq:ind-set-LP} with a greedy strategy to round a feasible solution to the LP~\eqref{eq:b-ind-set-LP}. 
Let $y$ be a
feasible (fractional) solution to the latter LP. We  define 
 $\bar{y}=y/b$, which  is a feasible solution to the independent set LP relaxation~\eqref{eq:ind-set-LP}.

 We build the solution  $T\sse [n]$ iteratively; initially $T=\emptyset$. For each iteration $k=1,\cdots b$, we perform the following steps:
\begin{enumerate}
    \item Consider the solution $\bar{y}$ {\em restricted} to $[n]\setminus T$. Since this is a feasible solution to the independent set LP~\eqref{eq:ind-set-LP}, we use the independent set rounding algorithm 
     to obtain an integral solution $S_k\sse [n]\setminus T$.
    \item Update $T\gets T\cup S_k$.
\end{enumerate}
As $\{S_k\}$ are disjoint subsets, $T=\cup_{k=1}^b S_k$ is a feasible integral solution to \eqref{eq:b-ind-set-LP}. 

We now analyze the reward of the solution $T$. 
For any subset $U\sse [n]$ let $Y(U):= \sum_{j \in U} r_j\cdot y_j$ be the LP-value restricted to $U$.  Consider the two cases:
\begin{itemize}
    \item Suppose $Y([n]\setminus T)\ge \frac12 \cdot Y([n])$ at the end of the algorithm. It follows that $Y([n]\setminus T)\ge \frac12 \cdot Y([n])$ in each iteration $k$.
    Consider  the LP solution  $\bar{y}$ {\em restricted} to $[n]\setminus T$ (in iteration $k$). 
    Since the rounding algorithm for the independent set LP relaxation has integrality gap $\rho$, 
     $$r(S_k)\ge \frac1\rho \sum_{j\in [n]\setminus T} r_j\cdot \bar{y}_j = \frac{1}{\rho b} Y([n]\setminus T) \ge \frac{Y([n])}{2\rho b}.$$
    Adding over all $b$ iterations, $r(T)=\sum_{k=1}^b r(S_k)\ge \frac{Y([n])}{2\rho }$.
    \item Suppose $Y([n]\setminus T)<\frac12 \cdot Y([n])$ at the end of the algorithm. Then, 
    $$r(T)\ge Y(T) = Y([n]) - Y([n]\setminus T)>\frac12 \cdot Y([n]).$$
\end{itemize}
In either case, we obtain that the algorithm's reward $r(T)\ge \frac{1}{2\rho}\cdot Y([n])$. This proves that the integrality gap of \eqref{eq:b-ind-set-LP} is at most $2\rho$.

\vspace*{2 mm}
\noindent
{\bf (ii) Randomized Rounding for large $\theta$.} \quad
Let $\tau$ denote $\frac{\theta}{2\log m}$. Consider the special case when $s_j \leq \tau$ for all $j \in [n]$. In this case, the LP relaxation~\eqref{eq:det_LP} is a special case of
{\em packing integer programs} (PIPs), studied  in  \cite{Srinivasan99}. Theorem~3.7 in  \cite{Srinivasan99} implies an $O(m^{1/P})$ integrality gap for the LP~\eqref{eq:det_LP}, where 
$$P=\frac{\theta}{\max_{j\in [n]} s_j} \ge \frac{\theta}{\tau}=2\log m.$$
Therefore, the LP~\eqref{eq:det_LP} has constant integrality gap in this special case.

\vspace*{2 mm}
\noindent
{\bf (iii) Geometric grouping for the general case. } \quad
 We now consider the LP \eqref{eq:det_LP} in the general setting.  Let $y$ be a fractional solution to this LP.  As define above,  $\tau:= \frac{\theta}{2\log m}$. We first partition the elements into groups based on their sizes as follows:
$$G_k := \left\{
\begin{array}{ll}
    \{j\in [n] : s_j < \tau\} & \mbox{ if }k=0 \\
     \{j\in [n] : 2^{k-1} \tau \le s_j < 2^k \tau\} & \mbox{ if }k\ge 1  
\end{array}\right. .$$
Note that the number of groups is $K=O(\log\log m)$ as $\max_j s_j\le \theta$. We handle each group separately, and pick the maximum reward solution across the $K$ groups. 

Consider a group $G_k, k\ge 1$. Consider the fractional solution $z$ defined as:
$$z_j = \left\{\begin{array}{ll}
    y_j/4 & \mbox{ if }j\in G_k \\
0     & \mbox{ otherwise}
\end{array}\right.$$
We claim that $z$ is a feasible solution to the LP~\eqref{eq:b-ind-set-LP} restricted to $G_k$, and a suitable value of $b$. Indeed, consider any  $L\in \cL$. Then, 
$$\sum_{j\in L}z_j \le \frac{1}{2^{k-1}\tau} \sum_{j\in G_k\cap L} s_j\cdot z_j =  \frac{1}{2^{k+1}\tau} \sum_{j\in G_k\cap L} s_j\cdot y_j\le \frac{\theta}{2^{k+1}\tau}\le \frac{c}{2} \le \lfloor c\rfloor,$$
where we have used the fact that $y$ is a feasible solution to \eqref{eq:det_LP}, and $c := 
\theta / \max_{j\in G_k} s_j \ge \max\{1 , \frac{\theta}{2^k\tau}\}.$
It follows that   $z$ is a feasible solution to the LP~\eqref{eq:b-ind-set-LP} where $b=\lfloor c\rfloor$. Hence, using the  rounding algorithm  for \eqref{eq:b-ind-set-LP} mentioned in the first step above, we obtain a solution $V_k\sse G_k$ with reward $$r(V_k)\ge \frac1{2\rho} \sum_{j \in G_k} r_j\cdot z_j \ge\frac{1}{8\rho} \sum_{j\in G_k} r_j y_j.$$
Moreover, for each $L\in \cL$, we have $|V_k\cap L|\le b$. Hence, for any $L\in\cL$,
$$\sum_{j\in V_k\cap L} s_j\le \left(\max_{j\in G_k} s_j \right) \cdot |V_k\cap L| =\frac{\theta}{ c}  \cdot |V_k\cap L| \le \frac{\theta b}{ c}\le \theta.$$
Thus, $V_k$ is a feasible integral solution to \eqref{eq:det_LP}.

Finally we consider the case $k=0$, i.e., the group $G_0$. As argued in the second step above, we obtain  a solution $V_0 \sse G_0$ with reward $r(V_0)\ge \frac{1}{\sigma}\cdot \sum_{j\in G_0} r_j y_j$ where $\sigma\ge 1$ is constant.
It follows that $V_0$ is an integral solution to \eqref{eq:det_LP} as well.

Finally, choosing the best solution from $\{V_k\}$ over all groups, we obtain reward at least
$$\max_k r(V_k) \ge \frac1K \sum_{k} r(V_k) \ge \frac{1}{K\cdot \max(8\rho, \sigma)}\sum_k \sum_{j\in G_k} r_j y_j = \frac{1}{K\cdot \max(8\rho, \sigma)} \sum_{j\in [n]} r_j y_j.$$
This proves that the integrality gap of \eqref{eq:det_LP} is $O(\rho \log\log m)$.
  \end{proof}

We now combine Theorem~\ref{thm:pack-indset} with known results on maximum weight independent sets for rectangles and fat objects, to prove their $\alpha$-\packable property.
\begin{corollary}\label{cor:rectangle-pack}
The set-system where tasks are $n$ axis-aligned rectangles in the plane and resources are all points in the plane, is $O((\log\log n)^2)$-\packable.
\end{corollary}
\begin{proof}
The weighted independent set problem for rectangles in the plane has an LP-based $O(\log\log n)$ approximation~\cite{ChalermsookW21}.  Combined with Theorem~\ref{thm:pack-indset} and the fact that the number of points $m$ can be ensured to be $poly(n)$ (see \S\ref{subsec:rect}), we obtain that the set-system is $O((\log\log n)^2)$-\packable.  
\end{proof}

\begin{corollary}\label{cor:disk-pack}
The set-system where tasks are $n$ disks (of arbitrary radii) in the plane and resources are all points in the plane, is $O(\log\log n)$-\packable.
\end{corollary}
\begin{proof}
There is 
an LP-based $O(u(n)/n)$-approximation algorithm for  weighted independent set 
 on set-systems where the  ``union complexity'' of $n$ objects is at most $u(n)$~\cite{ChanH12}.  See the survey~\cite{APS08} for more details on union complexity. The union complexity of disks (of arbitrary radii) is $O(n)$. So there is an LP-based $O(1)$-approximation algorithm for weighted independent set. Combined with Theorem~\ref{thm:pack-indset} and that $m=poly(n)$, the result follows.  
\end{proof}

\begin{corollary}\label{cor:triangle-pack}
The set-system where tasks are $n$ fat triangles in the plane and resources are all points in the plane, is $O(\log^* n \cdot  \log\log n)$-\packable.
\end{corollary}
\begin{proof}
The union complexity of fat triangles is $u(n)=O(n\, \log^* n)$ \cite{AronovBES14}. Using the result from \cite{ChanH12}, we obtain an LP-based $O(\log^* n)$-approximation algorithm for the weighted independent set problem. Using Theorem~\ref{thm:pack-indset} and that $m=poly(n)$, the result follows.  
\end{proof}

\end{document}